\documentclass[11pt,draftcls,onecolumn]{IEEEtran}

\usepackage{graphicx}
\usepackage{url}
\usepackage{amsmath,amssymb,amsfonts,comment}
\usepackage{amssymb,amsfonts,amscd,dsfont,mathrsfs}
\usepackage{graphicx,float,psfrag,epsfig}
\usepackage{wrapfig}
\usepackage{algorithmic}
\usepackage{rotating}
\usepackage{multirow}
\usepackage{booktabs}

\newtheorem{propo}{Proposition}[section]
\newtheorem{lemma}[propo]{Lemma}
\newtheorem{definition}[propo]{Definition}
\newtheorem{coro}[propo]{Corollary}
\newtheorem{thm}[propo]{Theorem}

\newtheorem{remark}[propo]{Remark}

\newcommand{\R}{\mathbb{R}}
\newcommand{\reals}{\mathbb{R}}
\newcommand{\ra}[1]{\renewcommand{\arraystretch}{#1}}

\def\prob{{\mathbb P}}

\def\R{\mathbb R}

\def\Tr{{\rm Tr}}
\def\id{{\mathbb I}}
\def\ones{{\mathds 1}}
\def\hd{\hat{d}}
\def\hD{\widehat{D}}
\def\hX{\widehat{X}}
\def\hx{\hat{x}}

\def\Ln{{L}}
\def\cL{{\cal L}}

\def\MDS{{\sf MDS }}
\def\<{\langle}
\def\>{\rangle}

\def\Orth{{\sf O}}

\def\d{{d_{\rm inv}}}
\def\tR{\widetilde{R}}
\def\tr{\tilde{R}}
\def\hopterrain{{\sc Hop-TERRAIN}}
\def\DVhop{{\sc DV-hop }}

\def\DVhopns{{\sc DV-hop}}
\def\RMDS{R_{\rm MDS}}
\def\RHOP{R_{\rm HOP}}

\title{Robust Localization from Incomplete Local  Information }
\author{{Amin Karbasi,~\IEEEmembership{Student Member, IEEE},
        Sewoong~Oh,~\IEEEmembership{Member,~IEEE}}
\thanks{{\footnotesize 
     Amin Karbasi is with the School of Computer and Communication
      Sciences, Ecole Polytechnique F\'{e}d\'{e}rale de Lausanne
      (EPFL), CH-1015 Lausanne, Switzerland (email: amin.karbasi@epfl.ch).
      Sewoong Oh is with the department of Electrical Engineering and Computer Science, Massachusetts Institute of Technology (MIT), Cambridge, MA 02139 (email: swoh@mit.edu).
  A preliminary summary of this work appeared  in \cite{OKM10, KO10}.   }} }

\begin{document}

%\markboth{A. Karbasi and S. Oh}{Positioning Through Noisy and  Local  Information in Ad-hoc Networks}

%\thanks{Manuscript received April 19, 2005; revised January 11, 2007.}}

\maketitle
\maketitle

% make the title area
%\maketitle

%
%============================================================
%
\begin{abstract}
We consider the problem of localizing wireless devices  
in an ad-hoc network embedded in a $d$-dimensional Euclidean space. 
Obtaining a good estimation of where wireless devices are located is 
crucial in wireless network applications including environment monitoring, geographic routing and topology control.
When the positions of the devices are unknown and 
only local distance information is given, 
we need to infer the positions from these local distance measurements. 
This problem is particularly challenging when 
we only have access to measurements that have limited accuracy and are incomplete. 
We consider the extreme case of this limitation on the available information, 
namely only the connectivity information is available, 
i.e., we only know whether a pair of nodes is within a fixed detection range of each other or not, 
and no information is known about how far apart they are.  
Further, to account for detection failures, 
we assume that even if a pair of devices is within the detection range, 
it fails to detect the presence of one another with some probability 
and this probability of failure depends on how far apart those devices are. 
Given this limited information, %the above set of noisy and incomplete measurements, 
we investigate the performance of a centralized positioning algorithm {\sc MDS-MAP} 
introduced by Shang et al. \cite{SRZ03}, 
and a distributed positioning algorithm {\hopterrain} introduced by Savarese et al. \cite{SLR02}. 
In particular, for a network consisting of $n$ devices positioned 
randomly, we provide a bound on the resulting error for both algorithms. 
We show that the error is bounded, decreasing at 
a rate that is proportional to $R_{\rm Critical}/R$, 
where $R_{\rm Critical}$ is the critical detection range when the resulting random network starts to be connected, 
and $R$ is the detection range of each device. 

\end{abstract}
%

%\category{C.2.1} {Computer-Communication Networks } {Network Architecture and Design} 

%\terms{Theory, Performance, Algorithms}

%\keywords{centralized, distributed, localization, sensor network}

%\acmformat{Karbasi, A., Oh, S.  2011. Positioning from Noisy and  Local Connectivity in Ad-hoc Networks.}

%\begin{bottomstuff}
%Author's addresses: A. Karbasi, School of Computer and Communication Sciences, Ecole Polytechnique Federal de Lausanne, Switzerland; S. Oh, Massachusetts Institute of Technology, Cambridge, MA 02139, USA.
%\end{bottomstuff}

%
%============================================================
%

\section{Introduction}
\label{sec:introduction}
In this paper, we address the problem of positoining (also referred to as sensor localization) 
when only a set of incomplete pairwise distances is provided. 
Location estimation  of individual nodes  is a requirement of 
many wireless sensor networks such as environment monitoring, 
geographic routing and topology control, to name only a few 
(for a thorough list of applications we refer the interested readers to \cite{survey,Xu02}). 
In  environment monitoring for instance, the measurement data by the wireless 
sensor network is essentially meaningless without knowing  from where the data is collected.

One way to acquire the positions is to equip all the sensors with 
a global positioning system (GPS). The use of GPS not only adds considerable cost to the system, 
but more importantly, it does not work in indoor environments  or when the received GPS signal is jammed (see \cite{gps-free} and the references therein for more information on this issue).
As an alternative, we seek an algorithm that can derive 
positions of sensors based on local/basic  information such as proximity 
(which nodes are within communication range of each other) 
or local distances (pairwise distances between neighbouring sensors).

Two common techniques for obtaining the local distance and connectivity information are Received Signal Strength Indicator (RSSI) and Time Difference of Arrival (TDoA). RSSI is a measurement of the ratio of the power present in a received radio signal and a reference power. Signal power at the receiving end is inversely proportional to the square of the distance between the receiver and the transmitter. Hence, RSSI has the potential to be used to estimate the distance and it is common to assume the use of RSSI in distance measurements. However, experimental results indicate that the accuracy of RSSI is limited \cite{TDoA}. TDoA technique uses the time difference between the receipt of two different signals with different velocities, for instance ultrasound and radio frequency signals \cite{TDoA2}. The time difference is proportional to the distance between the receiver and the transmitter; and given the velocity of the signals, the distance can be estimated from the time difference. These techniques can be used, independently or together, for distance estimation. In an alternative approach, Angle of Arrival (AoA) can also be used to infer the positions of sensors \cite{Niculescu2001}. 
Once a node has the angle of arrival information from three other nodes with known positions, we can perform triangulation to locate the wireless node. To measure the angle of arrival, an antenna array is required at each wireless node.

Given a set of  measurements, the problem of localization is solvable, meaning that it has a unique set of coordinates 
satisfying the given local information, only if there are enough constraints.  The simplest of such algorithms, i.e., multi dimensional scaling (MDS) \cite{mds-book}, assumes that all pairwise distances are known. Intuitively, it is clear that with $O(n^2)$ pairwise distances we should be able to determine $O(n)$ coordinates. However,  in almost all practical scenarios such information is unavailable for two major reasons. First,  sensors are typically highly resource-constrained (e.g., power) and have limited communication range. Thus, far away sensors  cannot communicate and obtain their pairwise distances. Second, due to noise and interference among sensors, there is always the possibility of non-detection or completely incoherent measurements.  

Many algorithms have been proposed to resolve these issues by using 
heuristic approximations to the missing distances, and  their success has 
mostly been measured experimentally. Regarding  the mechanisms deployed for estimating sensor locations, 
one can divide the localization algorithms into two categories: range-based and range-free. 
In the range-based protocols the absolute point-to-point distance estimates are used for 
inferring the locations, whereas in the range-free protocols no assumptions about the availability of 
such information are made and only the connectivity information is provided.  As a result, 
range-free algorithms are  more effective in terms of stability and cost, hence more 
favourable to be deployed in practical settings.

The theoretical guarantees associated with the performance 
of the existing methods are, however, of the same interest and complementary in nature.
%Of the same interest and complementary in nature, however, 
%are the theoretical guarantees associated with the performance 
%of the existing methods.
Such analytical bounds on the performance of localization algorithms 
can provide answers to practical questions: for example,"
How large should the radio range be in order to get the reconstruction error within a threshold?"
%understanding of theoretical promises (e.g., localizing the sensors up to bounded errors) is crucial in assessing the behavior of these methods. 
With this motivation in mind, our work takes a step forward in this direction. 

We first focus on providing a bound on the performance 
of a popular localization algorithm {\sc MDS-MAP} \cite{SRZ03}  when applied to
sensor localization from only connectivity information. 
We should stress here that pairwise distances are invariant under rigid transformations (rotation, translation and reflection). Hence, given  connectivity information, we can only hope to determine the \textit{configuration} or the relative map of the sensors. In other words,  localization is possible only up to rigid transformations.  With this point in mind,
we prove that using {\sc {\sc MDS-MAP}}, 
we are able to localize sensors up to a bounded error 
in a connected network where 
most of pairwise distances are missing and 
only local connectivity information is given. 

More precisely, assume that the network consists of $n$ sensors positioned randomly 
in a $d$-dimensional unit cube with the radio range $R=o(1)$ and detection probability $p$.
Let the $n\times d$ matrices $X$ and $\hat{X}$ denote the true sensor positions and their estimates by {\sc MDS-MAP}, respectively. Define $L=\id_{n\times n}-(1/n)\ones_n\ones_n^T$ where $\id_{n\times n}$ is the identity matrix and $\ones_n$ is the all ones vector. It is not difficult to show that $LXX^TL$ satisfies nice properties, specifically, it is invariant under rigid transformations and if $LXX^TL=L\hat{X}\hat{X}^TL$, then $X$ and $\hat{X}$ are equal up to rigid transformations. Therefore, we can naturally define the distance between $X$and $\hat{X}$ as follows: 
$$\d(X,\hX) = \frac{1}{n}\big\|\Ln XX^T\Ln-\Ln\hX\hX^T\Ln \big\|_F,$$
where $\|\cdot\|_F$ denote the Frobenius norm.  Our first result establishes a bound on the error of {\sc MDS-MAP} in terms of $\d$, specifically, $$\d(X,\hX)\leq \frac{\RMDS}{R}+o(1),$$ where $\RMDS= C_d(\ln(n)/n)^{1/d}$ for some constant $C_d$ that only depends on the dimension $d$.

One consequence of the ad-hoc nature of the underlying networks 
is the lack of a central infrastructure. This fact prevents the use of common 
centralized positioning algorithms such as {\sc MDS-MAP}. In particular, centralized algorithms suffer from the scalability problem, and generally it is not feasible for them to be implemented in large scale sensor networks. Other disadvantages of centralized algorithms, as compared to distributed algorithms, are their requirements for higher computational complexity and lower reliability; these drawbacks are due to accumulated information inaccuracies caused by multi-hop transmission over a wireless network \cite{survey2}.

We then investigate an important question about whether similar performance guarantees can be obtained in a distributed setting where each sensor tries to estimate its own \textit{global} position. As mentioned above, this task cannot be accomplished unless some additional information, rather than local measurements, is provided. It is well known that in a $d$-dimensional Euclidean space,  we need to know the global positions of at least $d+1$ sensors, referred to as \textit{anchors},  in order to uniquely determine the global positions of  the remaining sensors~\cite{Niculescu2001}. 

For the decentralized scenario, we turn our attention to analysing the performance of 
a popular localization algorithm called {\sc HOP-TERRAIN} algorithm  \cite{SLR02}. This algorithm  can be seen as a distributed version of the {\sc MDS-MAP}.
Similar to {\sc MDS-MAP},  we prove  that by using {\sc Hop-TERRAIN},
we are able to localize sensors up to a bounded error in a connected
network where most of the pairwise distances are unknown and 
only local connectivity information is given.

More formally, assume that  the network consists of $n$ sensors positioned randomly 	in a $d$-dimensional unit cube and $d+1$ anchors in general positions. Moreover, we let the radio range $R=o(1)$ and denote the detection probability by $p$. We show that when only connectivity information is available, for every unknown node $i$, the Euclidean distance between the estimate $\hat{x}_i$ and the correct position $x_i$ is bounded by 
\begin{displaymath}
\|x_i-\hat{x}_i\| \leq \frac{\RHOP}{R}+o(1),
\end{displaymath}
where $\RHOP=C'_d  (\log n/ n)^{\frac{1}{d}}$ for some constant $C'_d$ that only depends on $d$. 

\section{Related work}

The localization problem and its variants  has  attracted  significant research interest in recent years.
A general survey of the area and an overview of recent techniques can be found in \cite{Niculescu2001} and \cite{survey2}, respectively.  The problem is also closely related to dimensionality reduction \cite{Roweis2000} and manifold learning \cite{Saul2003} in which  the objects/data come from a high dimensional space, and the goal is to compute a low-dimensional,  neighbourhood preserving  embeddings.

In the case when all pairwise distances are known, 
the coordinates can be derived by using a classical method known as multidimensional scaling (MDS) \cite{mds-book}. 
The underlying principle of the MDS is to convert distances into an inner product matrix, 
whose eigenvectors are the unknown coordinates. 
In the presence of noise, MDS tolerates errors gracefully due to the overdetermined nature of the solution. 
However, when most pairwise distances are missing, the problem of finding the unknown 
coordinates becomes more challenging. For centralized algorithms (where all the measurements are sent to a single processor and the estimated positions are computed) three types of practical solutions to the above problem 
have been proposed in the literature. The first group consists of algorithms that try first to estimate 
the missing entries of the distance matrix and then apply MDS to the reconstructed 
distance matrix to find the coordinates of the sensors. {\sc MDS-MAP}, introduced in \cite{SRZ03} and further  studied in \cite{SRZ04}, 
can be mentioned as a well-known example of this class where it computes 
the shortest paths between all pairs of nodes in order to approximate 
the missing entries of the distance matrix. 
The algorithms in the second group mainly consider the sensor localization as 
a non-convex optimization problem and directly estimate the coordinates of sensors. 
A famous example of this type is a relaxation to semidefinite programming (SDP)\cite{BY04}. 
In the third group, the problem is formulated through a stochastic optimization where the main technique used in these algorithms is the stimulated annealing, which is a generalization of the Monte Carlo method in combinatorial optimization \cite{Kannan2006a,Kannan2006b}.

\begin{center}
\begin{table*}[t]
\caption{Distributed localization algorithm classification \cite{LR03} }
\begin{center}
\begin{tabular}{llll} 
\label{tab:3phase}
\vspace{-0.4cm}\\ 
\hline
Phase  & Robust positioning & Ad-hoc positioning & $N$-hop multilateration \\
\hline
1. Distance & \DVhop & Euclidean  & Sum-dist\\
2. Position & Lateration & Lateration  & Min-max\\
3. Refinement & Yes & No  & Yes \\
\hline
\end{tabular}
\end{center}
\end{table*}
\end{center}

%Recently, a number of localization algorithms have been proposed for sensor networks \cite{SLR02,NN03,NSB03,RD07,SPS02,BY04,SRZ04,Singer08}. 
%Based on the approach of processing the distance measurements, 
%these algorithms can be classified into two categories: centralized algorithms and distributed algorithms.
%In centralized algorithms, all the distance measurements are sent 
%to a single processor where the estimated positions are computed.
%Two well-known centralized localization algorithms are 
%multidimensional scaling (MDS) based approaches \cite{SRZ04} and 
%semidefinite programming (SDP) based algorithms \cite{BY04}.
%However, the centralized algorithms typically have 
%low energy efficiency and low scalability 
%due to dependency on a central processor and excessive communication overload.

Perhaps a more practical and interesting case is when there is no central infrastructure.
% explain - 3 phases + 3 algorithms.
\cite{LR03} identifies a common three-phase structure of 
three, popular, distributed sensor-localization algorithms,
namely robust positioning \cite{SLR02}, ad-hoc positioning \cite{NN03} and N-hop multilateration \cite{SPS02}.
Table \ref{tab:3phase} illustrates the structure of these algorithms.
%These algorithms have a common three phases where 
In the first phase, nodes share information to collectively determine
the distances from each of the nodes to a number of anchors. Anchors are
special nodes with a priori knowledge of their own position in some global coordinate system.
In the second phase, nodes determine their position based on 
the estimated distances to the anchors provided by the first phase 
and the known positions of the anchors. In the last phase, the initial estimated positions 
are iteratively refined. 
It is empirically demonstrated that 
these simple three-phase distributed sensor-localization algorithms 
are robust and energy-efficient \cite{LR03}. However, depending on which method is used in each phase, 
there are different tradeoffs between localization accuracy, computation complexity and power requirements.
In \cite{NSB03}, a distributed algorithm-called the Gradient algorithm-
was proposed; it is similar to ad-hoc positioning \cite{NN03} 
but uses a different method for estimating the average distance per hop.

%\begin{table}[h!b!p!]
%\caption{Distributed localization algorithm classification \cite{LR03} }
%\begin{tabular}{llll} 
%\label{tab:3phase}
%\vspace{-0.4cm}\\ 
%\hline
%Phase & Ad-hoc positioning \cite{NN03}& Robust positioning \cite{SLR02}& $N$-hop multilateration \cite{SPS02}\\
%Phase  & Robust positioning & Ad-hoc positioning & $N$-hop multilateration \\
%\hline
%1. Distance & \DVhop & Euclidean  & Sum-dist\\
%2. Position & Lateration & Lateration  & Min-max\\
%3. Refinement & Yes & No  & Yes \\
%\hline
%\end{tabular}
%\end{table}

Another distributed approach introduced in \cite{IFMW04} is to
pose the localization problem as an inference problem on a graphical model
and solve it by using Nonparametric Belief Propagation (NBP). 
It is naturally a distributed procedure and produces both an estimate of sensor locations 
and a representation of the location uncertainties. 
The estimated uncertainty may subsequently be used 
to determine the reliability of each sensor's location estimate.
% Our MDS-MAP paper
 
 The performances of these practical algorithms are invariably measured through simulations and 
little is known about the theoretical analysis supporting their results. 
A few exceptions are in the following work. In \cite{ DJM06} the authors use 
matrix completion methods \cite{Fazel02} as a means to reconstruct the distance matrix. 
The main contribution of their paper is that they are able to provably localize the sensors up to a bounded error. 
However, their analysis is based on a number of strong assumptions. 
First, they assume that even far-away sensors have a non-zero 
probability of detecting their distances. Second, the algorithm explicitly 
requires the knowledge of detection probabilities between all pairs. 
Third, their theorem only works when the average degree of the network 
(i.e., the average number of nodes detected by each sensor) 
grows linearly with the number of sensors in the network. 

Our first result, specifically the analysis of {\sc MDS-MAP}, has a similar flavour as in\cite{DJM06}. We provide a theoretical guarantee 
that backs up experimental results. We use shortest paths
as our primary guess for the missing entries in the distance matrix and apply MDS to find the topology of the network. 
In contrast to \cite{DJM06}, we require weaker assumptions for our results. 
More specifically, we assume that only neighbouring sensors have information about each other
and that only connectivity information is known. 
Furthermore, the knowledge of detection probabilities plays no role in our analysis or the algorithm. 
And last, in our analysis we assume that the average degree grows logarithmically-not linearly- with 
the number of sensors, which results in needing many less revealed entries in the distance matrix. In particular, the last condition is quite realistic: If the average degree grows
any slower then the network is not even connected (more on this issue in Section~\ref{sec:stretch}).
%This condition is completely justified, since the network will no longer be connected if the average degree does not grow logarithmically with the number of sensors. 
As the shortest paths algorithm works for both rage-free and range-aware cases, 
our analysis includes both and provides the first error bounds on the performance of {\sc MDS-MAP}. 

Of particular interest are the  two new results on the performance of sensor localization algorithms. In \cite{Javanmard2011}, Javanmard et al.  proposes a new reconstruction algorithm based on semidefinite programming where they could establish lower and upper bounds on the reconstruction errors of their algorithm. Similarly, in \cite{karbasi2010}, due to new advances in matrix completion methods \cite{Candes08}, the authors analyse the performance of OptSpace \cite{Keshavan2010}, a novel matrix completion algorithm, in localizing the sensors.  Interestingly, they did not need to adhere to the assumptions made by \cite{DJM06}. However, they have a restrictive assumption about the topology of the network, specifically, sensors are scattered inside an annulus.

%In an alternative line of work, the authors of \cite{OKM10} provided a bound on the performance of 
%a centralized localization algorithm known as {\sc MDS-MAP}.
%When $n$ sensors are randomly distributed in a unit $d$-dimensional hypercube,
%{\sc MDS-MAP} is able to localize sensors up to a bounded error with 
%only the pair-wise connectivity information, namely 
%whether two nodes are within a given radio range $r$ or not.

All the above analytical results  crucially rely on the fact that there is a central processor with 
access to the inter-sensor distance measurements.
However, as we have mentioned earlier,  centralized algorithms suffer from the scalability problem 
and require higher computational complexity.
Hence, a distributed algorithm with similar a performance bound is desirable.
%Can the above approach be generalized for a truly distributed sensor localization algorithm?
In our second result, we analyse the reconstruction error of a distributed sensor localization algorithm. To the best of our knowledge  
we show for the first time that {\sc Hop-TERRAIN}, introduced in \cite{SLR02}, achieves 
a bounded error when only local connectivity information is given.

Finally, one of the fundamental challenges in localization problem is whether, given a set of measurements, the sensor network is uniquely localizable or not. In the noiseless setting where all the measurements are accurate, it was shown that the correct notion through which we can answer this question is the global rigidity \cite{rigidity}, a property that is easy to check (a thorough discussion of global rigidity and its implications for the sensor localization problem is given in \cite{Gortlerabc}). However, finding such a unique solution is NP-hard \cite{Jackson2005}. In the case of noisy distance measurements very little is known in this area. For instance, we do not know  the fundamental limits for sensor localization algorithms or whether there are any algorithms with proven guarantees. From this point of view, our results narrow the gap between the algorithmic aspect of sensor localization and the theoretical one. In particular, we show that even in the presence of noise, the {\sc MDS-MAP} and {\sc HOP-TERRAIN} algorithms can localize the nodes within a bounded error. 
 
The organization of this paper is as follows. In Section~\ref{sec:model} we introduce the model and the notation used in our work. In Section~\ref{sec:algorithm} we describe the {\sc MDS-MAP} and {\sc HOP-TERRAIN} algorithms and their common features. Our results are stated in Section~\ref{sec:main} where we provide their proofs in Section~\ref{sec:mainproof}. Finally, we conclude in Section~\ref{sec:conclusion}.

%Section 2 describes a distributed sensor localization algorithm
%known as {\sc Hop-TERRAIN} \cite{SLR02} and 
%states the main results on the performance. 
%In Section 3, we provide detailed proof of the main theorems.

%
%============================================================
%
\section{Model definition}
\label{sec:model}

Before discussing the centralized and distributed localization algorithms in detail, 
we  define the mathematical model considered in this work.
First, we assume that we have no fine control over the placement of the sensors
that we call the {\em unknown nodes} (e.g., the nodes are dropped from an airplane).
%We assume that we have no fine control over the placement of the nodes 
%(e.g., the nodes are dropped from an airplane). 
Formally, we assume that $n$ nodes are placed uniformly at random 
in a $d$-dimensional cube $[0,1]^d$.

Additionally, we assume that there are $m$ special sensors, which we call {\em anchors},
with a priori knowledge of their own positions in some global coordinate. In practice, it is reasonable to assume that we have some control over the position of anchors. Basically, anchors are the nodes that are planted on the field before any positioning takes place. 

Let $V_a=\{1,\ldots,m\}$ denote the set of $m$ vertices corresponding to the anchors and 
$V_u=\{m+1,\ldots,m+n\}$ the set of $n$ vertices corresponding to the unknown nodes. We use $x_i$ to denote the random position of the node $i$ and $X$ to denote the $n\times d$ position matrix where the $i$-th row corresponds to $x_i$.

%Let $V=\{1,\ldots,n\}$ denote a set of $n$ vertices corresponding to the $n$ nodes. 
In positioning applications, due to attenuation and power constraints, only measurements between close-by nodes are available. 
As a result, the pairwise distance measurements can be represented by 
a random geometric graph $G(n+m,R)=(V,E,P)$, 
where $V=V_u \cup V_a$, $E\subseteq V\times V$ is a set of undirected edges 
that connect pairs of sensors that are close to each other, 
and $P\,:\,E\rightarrow \R^+$ is a non-negative real-valued function. 
The function $P$ is a mapping from 
a pair of connected nodes $(i,j)\in E$ to 
a distance measurement between $i$ and $j$.

A common model for this random geometric graph is the disc model 
where node $i$ and $j$ are connected if the 
Euclidean distance $d_{i,j}\equiv\|x_i-x_j\|$ is less than or equal to
a positive radio range $R$. 
In formulae,
\begin{equation*}
	(i,j)\in E \Leftrightarrow 	d_{i,j}\leq R \;. 
\end{equation*}
%In practice, due to signal obstruction and interference, 
%two nodes within the radio range $R$ might fail to detect each other. 
As mentioned earlier, there are a variety  of  ways
to measure the connectivity between two nodes, including time difference of arrival and RF received-signal strength (also called RF ranging). Due to limited resources, in all of the mentioned solutions there is  a probability of  non-detection (or completely wrong
estimation). Think of RF ranging in the presence of an obstacle
or in the (frequent) case of multiple paths. Depending on the
acquisition mechanism, this may result in the absence of measurement
or in incoherent measurements.

Throughout this paper, to model this failure of detection, we assume that two nodes can detect each other 
with a probability that only depends on the distance $d_{i,j}$.
Namely, $(i,j)\in E$ with probability $p(d_{i,j})$ if $d_{i,j}\leq R$.  
The detection probability 
$p(\cdot):\,[0,R]\rightarrow[0,1]$ is a non-increasing function of the distance. 
We consider a simple function parameterized by two scalar values $\alpha\in(0,1]$ and $\beta\in[0,3)$: 
%to simplify calculations we represent it with two parameters $\alpha$ and $\beta$. 
\begin{eqnarray}
	\label{eq:detection}
	p(z)=\min\left(1,\alpha \left(\frac{z}{R}\right)^{-\beta}\right) \;, 
\end{eqnarray}
for $\alpha\in(0,1]$ and $\beta\in[0,d)$.  Note that this includes the disc model with perfect detection as a special case (i.e., $\alpha=1,\beta=0$). 
\begin{figure}[t]
	\begin{center}
	\includegraphics[width=9cm]{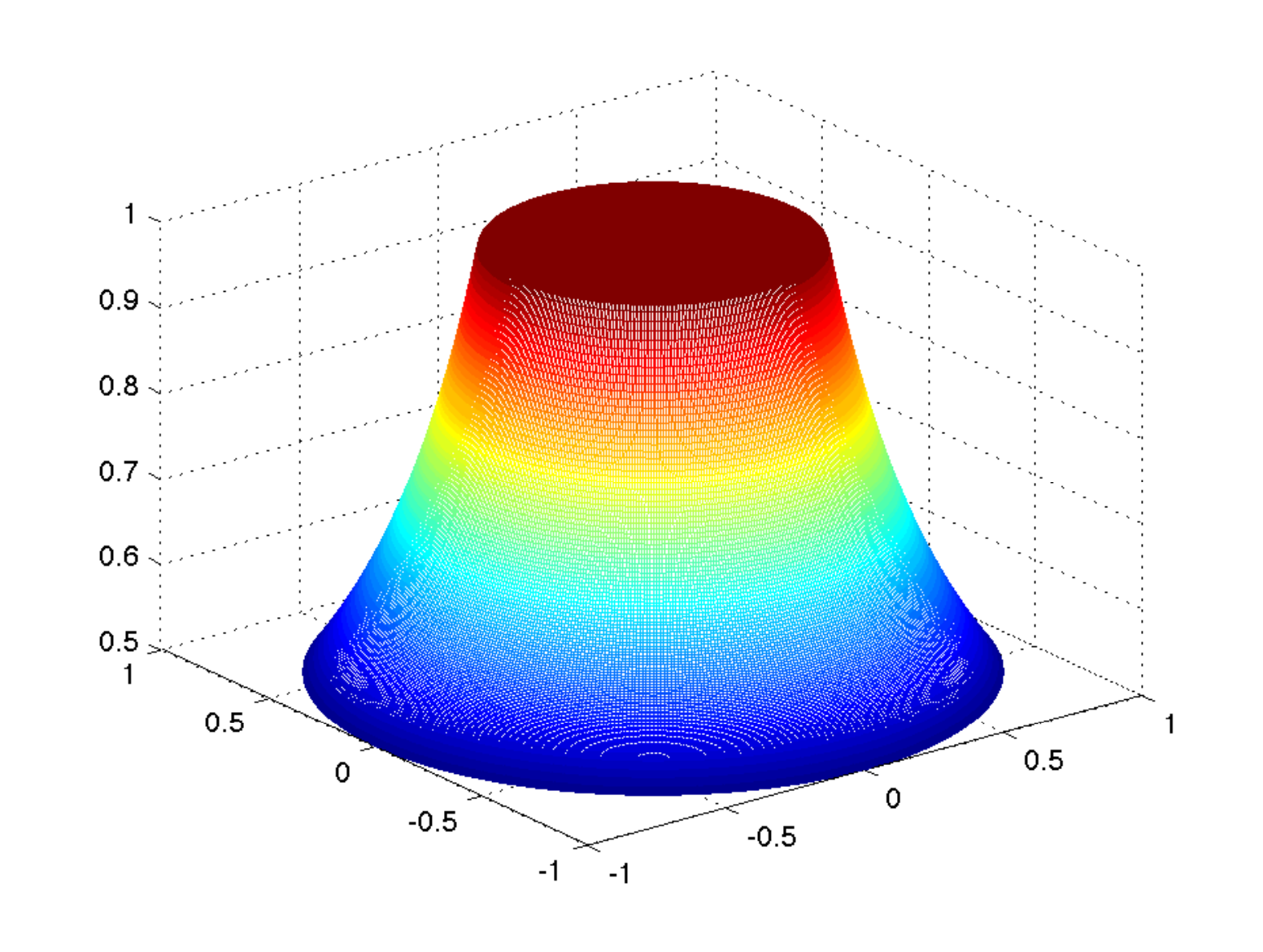}
	\put(-205,25){\footnotesize$x$}
		\put(-265,105){\footnotesize$p(z)$}
		\put(-75,17){\footnotesize$y$}
	\caption{This plot shows how the probability of detection changes  as the distance between two sensor changes.} 
	\label{fig:channel}
	\end{center}
\end{figure}
\begin{figure}[t]
	\begin{center}
	\includegraphics[width=9cm]{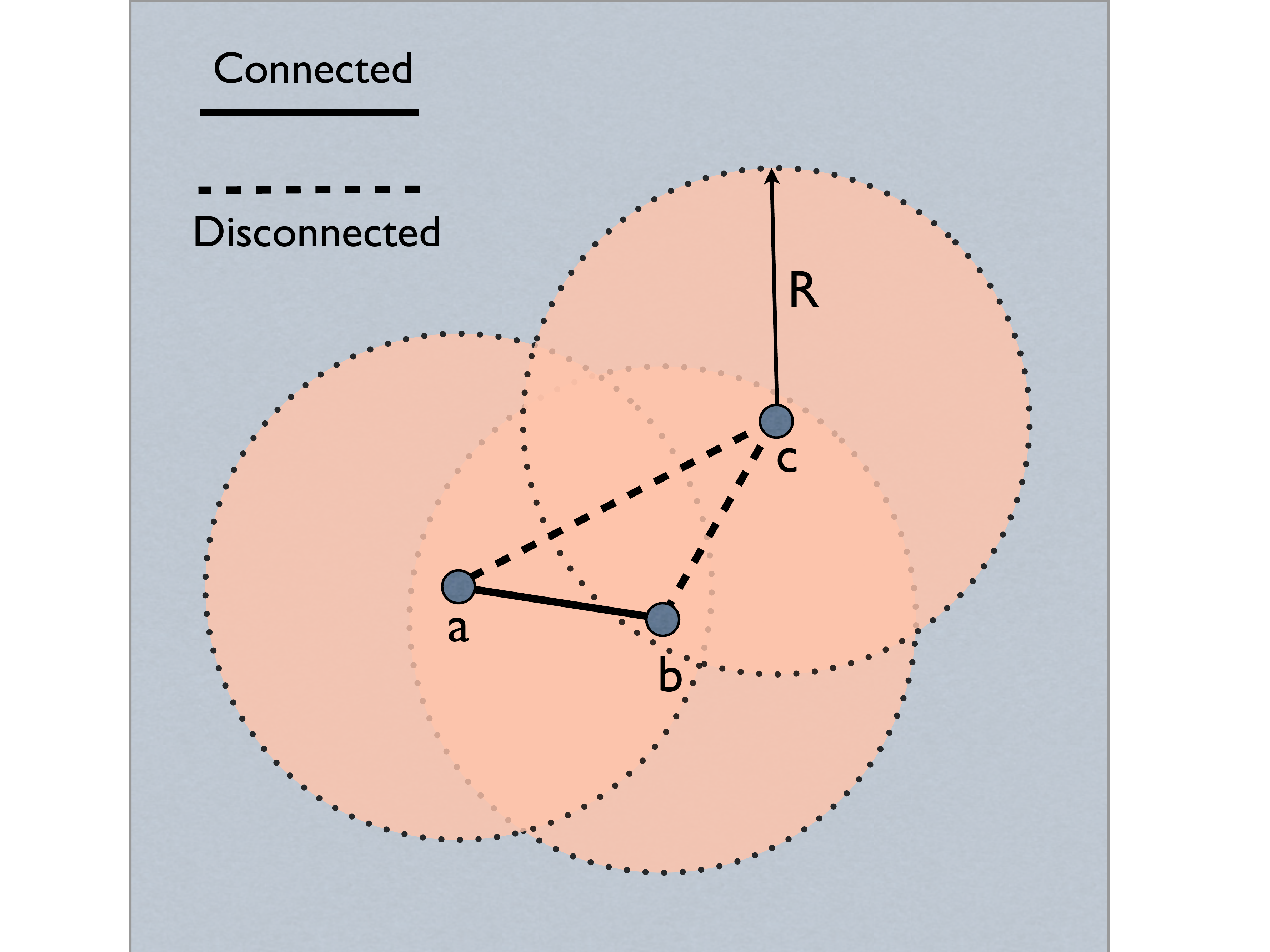}
	\caption{This example shows the model we consider in this work. Nodes $a$ and $b$ are connected since they are within radio range $R$ from each other. Even though the similar situation presents for $b$ and $c$, they are not connected due to detection failure. Finally, nodes $a$ and $c$ are not connected because they are far apart.} 
	\label{fig:RGG}
	\end{center}
\end{figure}

To each edge $(i,j)\in E$, we associate the distance measurement $P_{i,j}$
between sensors $i$ and $j$. 
In an ideal case, we have exact distance measurements available for those pairs in $E$. 
This is called the {\em range-based model} or the range-aware model. In formulae, 
\begin{equation*}
	P_{i,j}=\left\{  \begin{array}{ l l}
       	d_{i,j} &\quad \text{if } (i,j)\in E,    \\
       	* &\quad \text{otherwise},
	\end{array}
 	\right.
\end{equation*}
where a $*$ denotes that the distance measurement is unavailable.

In this paper, we assume that we are given only network connectivity information and 
no distance information. 
This is known as the {\em connectivity-based model} or the range-free model.  
More formally, 
\begin{equation*}
	P_{i,j}=\left\{  \begin{array}{ l}
       	1 \quad \text{if } (i,j)\in E,    \\
       	* \quad \text{otherwise}.
	\end{array}
 	\right.
\end{equation*}
In the following, let $D$ denote the $n\times n$ squared distance matrix where $D_{i,j}=d_{i,j}^2$. By definition, $$ D=a\ones_n^T+\ones_n a^T-2XX^T,$$ where $a\in \mathbb{R}^n$ is a vector with $a_i=\|x_i\|^2$ and $\ones_n$ is the all ones vector. As $D$ is a sum of two rank-1 matrices and a rank-$d$ matrix, its rank is at most $d+2$.
%\begin{remark}
%Although we assume the space to be $[0,1]^d$ hypercube for the sake of simplicity, our analysis and results easily generalize to any bounded convex set where the sensor nodes are distributed according to the homogeneous Poisson process with density $\rho=n$ (in only the unknown nodes are placed randomly) or  $\rho=n+m$ (if both the unknown and anchor nodes are placed randomly). More precisely, this is characterized by the probability that there are exactly $k$ nodes appearing in any region with volume $A$:$$\P(k_A=k)=\frac{(\rho A)^k}{k!} \exp(-\rho A).$$
%\end{remark}
\begin{table}[tb]
\caption{Summary of Notation.}
\ra{0.8}
\centering
\begin{tabular}{@{}ll | ll@{}}
\toprule
%Symbol & Meaning & Symbol & Meaning\\
%\midrule
$n$ & number of unknown sensors & $V_u$ & set of unknown nodes\\
$m$ & number of anchors & $V_a$ & set of anchors\\
$R$ & communication range & $\ones_n$ & all ones vector of size $n$\\	
$P_{i,j}$ & distance measurements & $\hD$ & estimated squared distance matrix\\
$d_{i,j}$ & Pairwise distance between nodes $i$ and $j$ & $\id_{n\times n}$ & ${n\times n}$ identity matrix\\
$x_i$ & position of node $i$ & $\hat{x}_i$ & estimated position of node $i$\\
$p$ & detection probability  & $X$ &positions matrix\\
$d$ & dimension & $\hX$ &estimated positions matrix\\
$D$ & squared distance matrix & $\hd_{i,j}$& shortest path between node $i$ and $j$\\
$\Orth(d)$ & orthogonal group of $d\times d$ matrices & $\|\cdot \|_F$& Frobenius norm \\
$\langle A,B\rangle$ & Frobenius inner product & $\|\cdot \|_2$& spectral norm \\
\bottomrule
\end{tabular}
\label{tab:notations}
\end{table}

%
%============================================================
%
\section{Algorithms}
\label{sec:algorithm}
%In this section we briefly  explain the centralized algorithm
%\\{\sc MDS-MAP} and its distributed counterpart HOP-TERRAIN.

In general, there are two solutions to the localization problem: a relative map and an absolute map.
%Note that when solving a localization problem, 
%there are two possible outputs : a relative map and an absolute map.
A relative map is  a 
configuration of sensors that have the same neighbor 
relationships as the underlying graph $G$. 
In the following we use the terms 
configuration, embedding, and relative map interchangeably.
 An absolute map, on the other hand,
determines the absolute geographic coordinates of all sensors. 
In this paper our objective is two-fold. First, we present the centralized algorithm {\sc {\sc MDS-MAP}}, that
%In this paper, we are interested in presenting the centralized algorithm {\sc \{\sc MDS-MAP}}, that 
finds a configuration that best fits the proximity measurements. Then,  we discuss its distributed version {\sc HOP-TERRAIN}  where its goal is for each sensor to find its absolute position.  For both,  we provide analytical bounds on the error between 
the estimated configuration and the correct configuration. 
%and providing an analytical bound on the error between 
%the estimated configuration and the correct configuration. Wit the same objective in mind, 
%Define a set of random positions of $n$ sensors ${\cal X}=\{x_1,\ldots,x_n\}$, 

%%%%%%%%%%%%%%%%%%%%%%%%%%%%%%%%%%%%%%%%%%%%%%%%%%%%%%%%%%%%%%%%%%%%%%%%%%%%%%%%%%%%%%%%%%%%%%%%%%%%%%%%

\subsection{Centralized Positioning Algorithm: {\sc MDS-MAP}}
\label{sec:centralalgorithm}
For the centralized positioning algorithm, we assume that there is no anchor node in the system, namely, $V_a=\phi$. We define a set of random positions of $n$ sensors ${\cal X}=\{x_1,\ldots,x_n\}$. 
{\sc MDS-MAP} consists of two steps: 

\begin{center}
\begin{tabular}{ll} 
\hline
%\vspace{-0.4cm}\\
\multicolumn{2}{l}{ {\bf Algorithm :} {\sc MDS-MAP} \cite{SRZ03} }\\
\hline
%\vspace{-0.4cm}\\
\multicolumn{2}{l}{{\bf Input:} dimension $d$, graph $G=(V,E,P)$ }\\
%\multicolumn{2}{l}{{\bf Output:} estimation $\hX$}\\
1: & Compute the shortest paths, and let $\hD$ be \\
    & the squared shortest paths matrix;\\
2: & Apply MDS to $\hD$, and let $\hX$ be the output.\\
%3: & Use anchor positions to find the scaling parameters.\\
\hline
\end{tabular}
\end{center}
%
\begin{comment}
The original {\sc MDS-MAP} algorithm introduced in \cite{SRZ03} 
has additional post-processing step to fit the given configuration 
to an absolute map using a small number (typically $d+1$) of special nodes 
with known positions, which are called anchors. 
However, the focus of this paper is to give a bound on the error between 
the relative map found by the algorithm and the correct configuration.
Hence, the last step, which does not improve the performance, is omitted here. 
\end{comment}

{\bf Shortest paths.}
 The shortest path between nodes $i$ and $j$ in graph $G=(V,E,P)$
is defined as a path between two nodes such that the 
sum of the proximity measures of its constituent edges is minimized.
Let $\hd_{i,j}$ be the computed shortest path between node $i$ and $j$. 
Then, the squared shortest paths matrix $\hD\in\R^{n\times n}$ is defined as 
$\hD_{ij}=\hd_{i,j}^2$ for $i\neq j$, and $0$ for $i=j$.

\begin{figure}[t]
\begin{center}
\vspace{-.75cm}
%\hspace{-1.5cm}
\includegraphics[width=12cm]{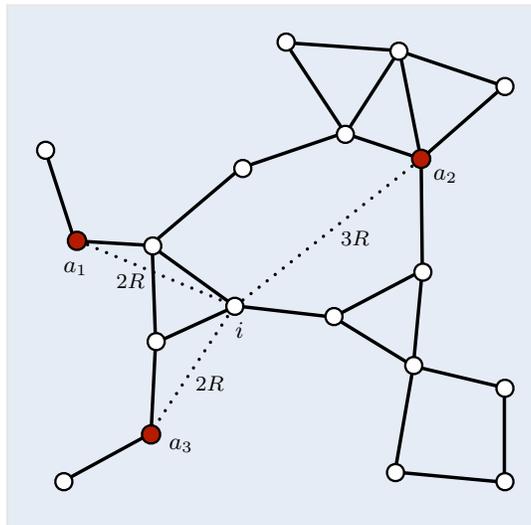}
\put(-185,105){\footnotesize$i$}
		\put(-250,130){\footnotesize$a_1$}
		\put(-110,165){\footnotesize$a_2$}
		\put(-210,63){\footnotesize$a_3$}
		\put(-145,140){\scriptsize$3R$}
		\put(-200,85){\scriptsize$2R$}
		\put(-230,124){\scriptsize$2R$}
		\vspace{-1cm}
\caption{The shortest path between two nodes is defined in terms of the minimum number of hops multiplied by the radio range $R$. For instance, the shortest path between $i$ and $a_1$ is $2R$. } \label{fig:shortest_path}
\end{center}
\end{figure}
{\bf Multidimensional scaling.}
In step 2, we apply the MDS to $\hD$ to get a good estimate of $X$, specifically, we compute $\hX=\MDS_d(\hD)$.
Multidimensional scaling (MDS) refers to a set of statistical techniques used in finding the configuration of objects in a low dimensional space such that the measured pairwise distances are preserved \cite{mds-book}. It is often used for a visual representation of the proximities between a set of items. For example, given a matrix of perceived similarities or dissimilarities between 
$n$ items, MDS geometrically places each of those items in a low dimensional space such that the items that are similar are placed close to each other. Formally, MDS finds a lower dimensional embedding $\hat{x_i}$s that minimize the \textit{stress} defined as
$$\text{stress}\doteq \sqrt{\frac{\sum_{i\neq j} (f(d_{i,j})-\hat{d}_{i,j})^2}{\sum{i\neq j}\hat{d}^2_{i,j}}},$$
where $d_{i,j}$ is the input similarity (or dissimilarity), $\hat{d}_{i,j}=\|\hat{x_i}-\hat{x_j}\|$ is the Euclidean distance in the lower dimensional embedding, and $f(\cdot)$ is some function on the input data. When MDS perfectly embeds the input data, we will have $f(d_{i,j})=\hat{d}_{i,j}$ and the stress is zero. 

In this chapter we use what is called the classic metric MDS (we refer the interested reader to \cite{Cox01}, for the definition of other types of MDS algorithms, for instance non-metric MDS, replicated MDS, and weighted MDS). 	In classic metric MDS, $f(\cdot)$ is the identity function and the input dissimilarities correspond to the Euclidean distances such that $d_{i,j}=\| x_i-x_j\|$ for some lower dimensional embedding $\{x_i\}$. Further, when all the dissimilarities (or pairwise distances) are measured without error, the following spectral method correctly recovers the lower dimensional embedding up to a rigid motion.

%MDS is a technique 
%used in finding the configuration of objects in a low dimensional space
%such that the measured pairwise distances are preserved.
%If all the pairwise distances are measured without error then 
%a naive application of MDS exactly recovers the configuration of sensors \cite{SRZ03}.

%
\begin{center}
\begin{tabular}{ll} 
\hline
%\vspace{-0.4cm}\\
\multicolumn{2}{l}{ {\bf Algorithm :} Classic Metric MDS \cite{SRZ03} }\\
\hline
%\vspace{-0.4cm}\\
\multicolumn{2}{l}{{\bf Input:} dimension $d$, estimated distance matrix $M$ }\\
%\multicolumn{2}{l}{{\bf Output:} estimated positions $\MDS_d(M)$}\\
1: & Compute $(-1/2)\Ln M\Ln$, \\
   & where $\Ln=\id_n-(1/n)\ones_n\ones_n^T$;\\
2: & Compute the best rank-$d$ approximation $U_d\Sigma_d U_d^T$ \\
   & of $(-1/2)\Ln M\Ln$;\\
3: & Return $\MDS_d(M)\equiv U_d\Sigma_d^{1/2} $. \\
%3: & Use anchor positions to find the scaling parameters.\\
\hline
\end{tabular}
\end{center}
This algorithm has been frequently used in positioning applications; and in the future, whenever we say MDS we refer to the above algorithm.
% Formal definition of MDS
%There are many types of MDS techniques, but, throughout this paper, 
%whenever we say MDS we refer to the classical metric MDS, which is defined as follows.
Let $\Ln$ be an $n \times n$ symmetric matrix such that 
$$\Ln=\id_n-(1/n)\ones_n\ones_n^T,$$
where $\ones_n \in \R^n$ is the all ones vector and 
$\id_n$ is the $n \times n$ identity matrix.
Let ${\sf MDS}_d(D)$ denote the $n\times d$ matrix 
returned by MDS when applied to the squared distance matrix $D$.
%Then the classical metric MDS of an $n\times n$ squared distance matrix $D$
%computes the first $d$ principal components of $(-1/2)\Ln D\Ln$. 
Then, in formula, given the singular value decomposition (SVD) 
of a symmetric and positive definite matrix $(-1/2)\Ln D \Ln$ 
as $(-1/2)\Ln D \Ln = U\Sigma U^T$, 
\begin{eqnarray*}
 \MDS_d(D) \equiv U_d \Sigma_d^{1/2}\;,
\end{eqnarray*}
where $U_d$ denotes the $n\times d$ left singular matrix that
corresponds to the $d$ largest singular values and $\Sigma_d$
denotes the $d\times d$ diagonal matrix with $d$ largest singular values in the diagonal.
This is also known as the {\sc MDSLocalize} algorithm in \cite{DJM06}.
Note that as the columns of $U$ are orthogonal to $\ones_n$ by construction, 
it follow that $\Ln\cdot\MDS_d(D) = \MDS_d(D)$.

%In other words, $\hX\hX^T$ is the rank-$d$ projection of $(-1/2)\Ln D\Ln$,
%where we define the rank-$d$ projection of a matrix $A$ as  
%$\cP_d(A)=\sum_{i=1}^d{\sigma_iu_iv_i^T}$.
%Here, $u_i$ and $v_i$ are the left and right singular vectors of $A$, respectively, 
%corresponding to the $i$th singular value $\sigma_i$.

%Property of MDS
It can be easily shown that 
when MDS is applied to the correct squared distance matrix without noise, 
the configuration of sensors are exactly recovered \cite{DJM06}.
This follows from the following equality
\begin{eqnarray}
 -({1}/{2})\Ln D \Ln = \Ln XX^T\Ln\;. \label{eq:LDL}
\end{eqnarray}
Note that we only obtain the configuration and not the absolute positions, 
in the sense that $\MDS_d(D)$ is one version of infinitely many solutions 
that matches the distance measurements $D$. 
%which are equal up to rigid transformation (rotations, reflections and translations). 
%Multiplicity
%The squared distance matrix $D$ is a function of the sensor positions $X$
%such that $D(X)_{ij}=\|x_{i}-x_{j}\|^2$, where $x_i$ and $x_j$ denote 
%the $i$th and $j$th rows of $X$, respectively.
%Intuitively, it is clear that the pairwise distances are invariant to 
%a rigid transformation (a combination of 
%rotation, reflection and translation) of the positions $X$, 
 Therefore there are multiple incidents of $X$ that result in the same $D$.
 We introduce a formal definition of rigid transformation 
and related terms.

We denote by $\Orth(d)$ the orthogonal group of $d\times d$ matrices.
A set of sensor positions $Y\in\R^{n\times d}$ is a rigid transformation of $X$, 
if there exists a $d$-dimensional shift vector $s$ 
and an orthogonal matrix $Q\in\Orth(d)$ such that 
%\begin{eqnarray*}
 $Y = XQ+\ones_n s^T \;$.
%\end{eqnarray*}
Here $Y$ should be interpreted as 
a result of first rotating (and/or reflecting) sensors in position $X$ by $Q$ and then 
adding a shift by $s$. 
%Note that we could have shifted the 
%sensors by another shift vector $s'$ and then rotated it 
%to get the same positions $Y=(X+\ones_n s'^T)Q$.
Similarly, when we say two position matrices $X$ and $Y$ 
are equal up to a rigid transformation, 
we mean that there exists a rotation $Q$ and a shift $s$ such that 
$Y = XQ+\ones_n s^T$.
Also, 
%With this definition of rigid transformation, 
%we can define what it means for a function to be invariant under rigid transformation. 
we say a function $f(X)$ is \textit{invariant} under rigid transformation 
if and only if for all $X$ and $Y$ that are equal up to a 
rigid transformation we have $f(X)=f(Y)$.
Under these definitions, it is clear that $D$ 
is invariant under rigid transformation, as for all $(i,j)$, 
$$D_{ij} = \|x_i-x_j\|^2 = \|(x_iQ+s^T)-(x_jQ+s^T)\|^2,$$ 
for any $Q\in\Orth(d)$ and $s\in\R^d$. 
\begin{comment}
Note that when solving a localization problem, 
there are two possible outputs : a relative map (or a configuration) and an absolute map.
The task of finding a relative map is to find a 
configuration of sensors that have the same neighbour 
relationships as the underlying graph $G$. 
In the following we use the terms 
configuration, embedding, and relative map interchangeably.
The task of finding an absolute map is to 
determine the absolute geographic coordinates of all sensors. 
In this paper, we are interested in presenting an algorithm {\sc MDS-MAP}, that 
finds a configuration that best fits the proximity measurements, 
and providing an analytical bound on the error between 
the estimated configuration and the correct configuration.
\end{comment}

Although MDS works perfectly when $D$ is available, 
in practice not all proximity measurements are available
because of the limited radio range $R$. This is why, in the first step, we estimated the unavailable entries of $D$ by finding the shortest path between disconnected nodes.

\subsection{Distributed Positioning Algorithm: {\sc HOP-TERRAIN}}
\label{sec:distributedalgorithm}
Recall that {\sc HOP-TERRAIN} is a distributed  algorithm that aims at finding the global map. Notice that  in order to fix the global coordinate system in a $d$ dimensional space, we need to know the positions of at least $d+1$ nodes. As we defined before, these nodes  whose global positions are known are called anchors. In this section we assume that we have $m$ anchors in total, i.e., $V_a=\{1,2,\dots,m\}$.  Based on the robust positioning algorithm introduced in \cite{SLR02}, 
the distributed sensor localization algorithm consists of two steps : 

\begin{center}
\begin{tabular}{ll} 
\hline
%\vspace{-0.4cm}\\
\multicolumn{2}{l}{ {\bf Algorithm :} \hopterrain \cite{SLR02} }\\
\hline
%\vspace{-0.4cm}\\
%\multicolumn{2}{l}{{\bf Input:} dimension $d$, graph $G=(V,E,P)$, anchor positions }\\
%\multicolumn{2}{l}{{\bf Output:} estimation $\{\hx_i\}$ }\\
1: & Each node $i$ computes the shortest paths \\
   & $\{\hd_{i,a}\,:\,a\in V_a\}$ between itself and the anchors;\\
2: & Each node $i$ derives an estimated position $\hx_i$ \\
   & by triangulation with a least squares method. \\
%3: & Use anchor positions to find the scaling parameters.\\
\hline
\end{tabular}
\end{center}
%
\begin{comment}
According to the three phase classification presented in Table \ref{tab:3phase}, 
this is closely related to the first two phases of the robust positioning algorithm.
This algorithm uses a slightly different method for computing the shortest paths, 
which is compared in detail later in this section. 
Hence, through out this paper, we refer to this algorithm as {\sc Hop-TERRAIN},
which denotes the first two steps of robust positioning algorithm in \cite{SLR02}.
\end{comment}

{\bf Distributed shortest paths:} Similarly to {\sc MDS-MAP}, the first step is about finding the shortest path. The difference is that in the first step each of the unknown nodes only estimates
the distances between itself and the anchors. 
These approximate 
distances will be used in the next triangulation step to derive an estimated position. In other words,
the shortest path between an unknown node $i$ and an anchor $a$ in the graph $G$
provides an estimate for the Euclidean distance $d_{i,a}=\|x_i-x_a\|$.
\begin{comment}
This approximate 
distances will be used in the next triangulation step to derive an estimated position.
The shortest path between an unknown node $i$ and an anchor $a$ in the graph $G$
provides an estimate for the Euclidean distance $d_{i,a}=\|x_i-x_a\|$, 
and for a carefully chosen radio range $R$ this shortest path estimation is close
to the actual distance as will be shown in Lemma \ref{lem:shortestpath}. 
\end{comment}

% Formally, the shortest path between an unknown node $i$ and an anchor $a$ in the graph $G=(V,E,P)$
%is defined as a path between two nodes such that the 
%sum of the proximity measures of its constituent edges is minimized.
We denote by $\hd_{i,a}$ the computed shortest path 
and this provides the initial estimate for the distance between the node $i$ and the anchor $a$. 
When only the connectivity information is available and the corresponding graph $G=(V,E,P)$ is defined as in the {\em connectivity-based model}, the shortest path $\hd_{i,a}$ is equivalent to 
the minimum number of hops between a node $i$ and an anchor $a$ multiplied by the radio range $R$. 

In order to find the minimum number of hops from  
an unknown node $i\in V_u$ to an anchors $a\in V_a$ in a distributed way, 
we use a method similar to \DVhop \cite{NN03}.
Each unknown node maintains a table $\{x_a,h_a\}$ that is initially empty, 
where $x_a\in\R^d$ refers to the position of the anchor $a$ and $h_a$ to 
the number of hops from the unknown node to the anchor $a$.
%Note that this algorithm is distributed and each node exchanges updates only with its neighbors.
First, each of the anchors initiate a broadcast containing 
its known location and a hop count of one.
All of the one-hop neighbors surrounding the anchor, 
on receiving this broadcast, record the anchor's position and a hop count of one,
and then broadcast the anchor's known position and a hop count of two. 
From then on, whenever a node receives a broadcast, it does one of the two things. 
If the broadcast refers to an anchor that is already in the record
and the hop count is larger than or equal to what is recorded, 
then the node does nothing. Otherwise, if 
the broadcast refers to an anchor that is new or has a hop count that is smaller, 
the node updates its table with this new information on its memory 
and broadcasts the new information after incrementing the hop count by one. 

To estimate the distances between the node and the anchors,
when every node has computed the hop count to all the anchors, 
the number of hops is multiplied by the radio range $R$ 
to estimate the distances between the node and the anchors.
Note that to begin triangulation, not all the hop counts to all the anchors
are necessary. A node can start triangulation as soon as it has 
estimated distances to $d+1$ anchors. 
There is an obvious trade-off between
the number of communications and their performance.

The above step of computing the minimum number of hops is 
the same distributed algorithm as described in \DVhopns. However, one difference is that 
instead of multiplying the number of hops by a fixed radio range $R$, 
in \DVhopns, the number of hops is multiplied by an average hop distance.
The average hop distance is computed from the known pairwise distances between anchors 
and the number of hops between the anchors.
although numerical simulations show that the average hop distance provides a better estimate,
the difference between the computed average hop distance and the radio range $R$ becomes negligible as $n$ grows large.

{\bf Triangulation using least squares.}
In the second step, each unknown node $i\in V_u$ uses 
a set of estimated distances $\{\hd_{i,a}:a\in V_a\}$ 
together with the known positions of the anchors, to perform 
a triangulation. The resulting estimated position is denoted by $\hx_i$. 
For each node, the triangulation consists in solving 
a single instance of a least squares problem ($Ax=b$) 
and this process is known as Lateration \cite{SRB01,LR03}.

For an unknown node $i$, the position vector $x_i$ and the anchor 
positions $\{x_a:a\in \{1,\ldots,m\}\}$ satisfy the following series of equations:
\begin{eqnarray*}
 \|x_1-x_i\|^2 &=& d_{i,1}^2 \;,\\
 &\vdots& \\
 \|x_m-x_i\|^2 &=& d_{i,m}^2 \;.
\end{eqnarray*}
Geometrically, the above equalities simply say that the point $x_i$ is the intersection point of $m$ circles centred at $x_1, x_2, \dots, x_m$ (see Figure~\ref{fig:lateration}).  
\begin{figure}[t]
\begin{center}
%\hspace{-2.1cm}
\includegraphics[width=7cm]{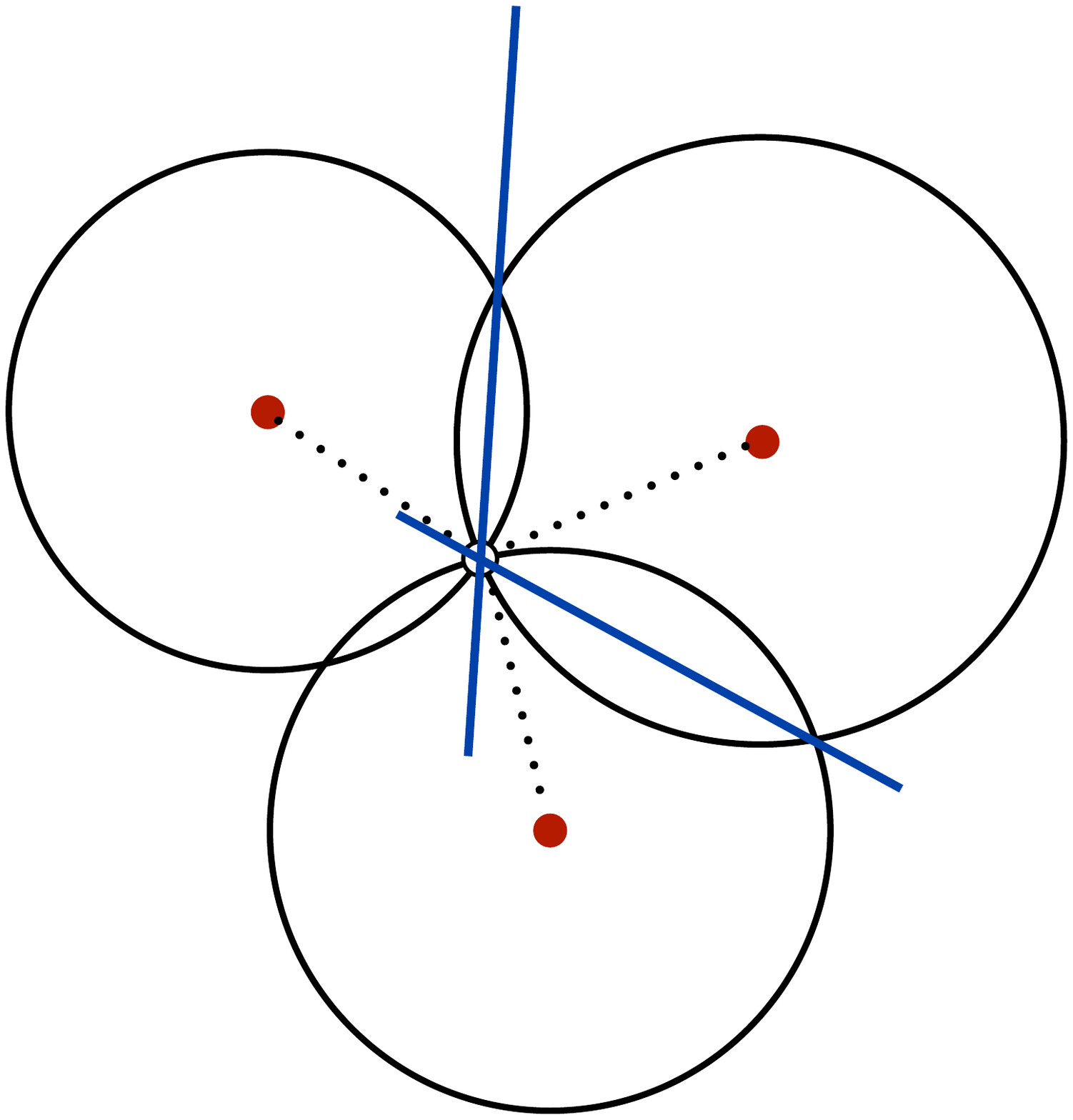}
\put(-120,50){$x_1$}
\put(-150,90){$x_2$}
\put(-80,90){$x_3$}
\vspace{-1cm}
\includegraphics[width=7cm]{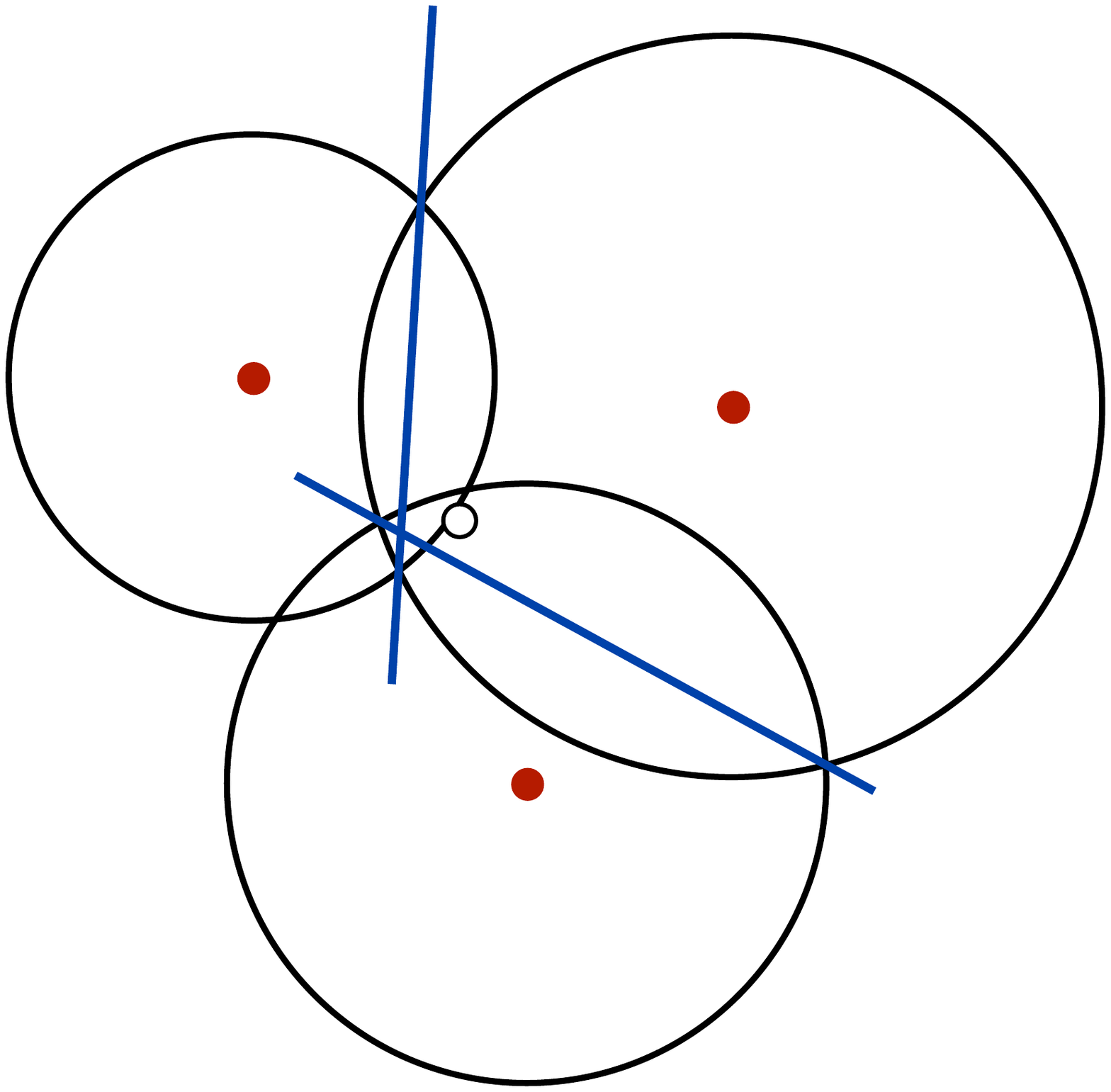}
\put(-120,50){$x_1$}
\put(-150,90){$x_2$}
\put(-80,90){$x_3$}
%\vspace{-.5cm}
\caption{Multilateration with exact distance measurements (left) and with approximate distance measurements (right). Three solid circles denote the anchors (red) and the white circle denotes the unknown nodes. The intersection of the blue lines corresponds to the solution of multilateration.} \label{fig:lateration}
\end{center}
\end{figure}
This set of equations can be linearised by subtracting each line from the next line.
\begin{eqnarray*}
 \|x_2\|^2-\|x_1\|^2+ 2(x_1-x_2)^Tx_i &=& d_{i,2}^2-d_{i,1}^2 \;,\\
 & \vdots & \\
 \|x_m\|^2-\|x_{m-1}\|^2+ 2(x_{m-1}-x_m)^Tx_i &= &d_{i,m}^2-d_{i,{m-1}}^2 \;.
\end{eqnarray*}

\begin{comment}
\begin{align*}
 &\|x_2\|^2-\|x_1\|^2+ 2(x_1-x_2)^Tx_i \\
 &\quad = d_{i,2}^2-d_{i,1}^2 \;,\\
 &\quad \quad \vdots \\
 &\|x_m\|^2-\|x_{m-1}\|^2+ 2(x_{m-1}-x_m)^Tx_i \\
 &\quad = d_{i,m}^2-d_{i,{m-1}}^2 \;.
\end{align*}
\end{comment}

By reordering the terms, we get a series of linear equations for node $i$ 
in the form $A\,x_i=b^{(i)}_0$,
for $A\in\R^{(m-1)\times d}$ and $b\in\R^{m-1}$ defined as
\begin{eqnarray*}
 A &\equiv& \begin{bmatrix}
	2(x_{1}-x_{2})^T\\
	\vdots \\
	2(x_{m-1}-x_{m})^T 
	\end{bmatrix} \;,\\
b_0^{(i)} &\equiv&  \begin{bmatrix}
	\|x_{1}\|^2-\|x_{2}\|^2  + d_{i,{2}}^2-d_{i,{1}}^2\\
	\vdots \\
	\|x_{m-1}\|^2-\|x_{m}\|^2 + d_{i,{m}}^2-d_{i,{m-1}}^2
	\end{bmatrix} \;.\\
\end{eqnarray*}
Note that the matrix $A$ does not depend on the particular unknown node $i$ and 
all the entries are known accurately to all the nodes after the distributed shortest paths step.
However, the vector $b^{(i)}_0$ is not available at node $i$, because 
$d_{i,a}$'s are not known. 
Hence we use an estimation $b^{(i)}$, that is defined from $b^{(i)}_0$
by replacing $d_{i,a}$ by $\hd_{i,a}$ everywhere. Notice that $\hd_{i,a}\geq d_{i,a}$. As a result, the circles centred at $x_1,x_2,\dots,x_m$ have potentially larger radii. Therefore,  the intersection between circles is no longer a single point, but rather a closed area. 
Then, finding the optimal estimation $\hx_i$ of $x_i$ that minimizes the mean squared error 
is solved in a closed form using a standard least squares approach: 
\begin{eqnarray}
 \hx_i=(A^TA)^{-1}A^Tb^{(i)} \;. \label{eq:lateration}
\end{eqnarray}

For bounded $d=o(1)$, a single least squares operation has complexity $O(m)$, 
and applying it $n$ times results in the overall complexity of $O(n\,m)$.
No communication between the nodes is necessary for this step.
%%%%%%%%%%%%%%%%%%%%%%%%%%%%%%%%%%%%%%%%%%%%%%%%%%%%%%%%%%%%%%%%%%%%%%%%%%%%%%%%%%%%%%%%%%%
\subsection{Stretch Factor: Euclidean Distance versus Shortest Path}\label{sec:stretch}

In general when the graph $G$ is not connected, the localization problem is not well defined. In fact, there are multiple configurations resulting 
in the same observed proximity measures. For instance if graph $G$ consists of two disconnected components, they can be placed in possibly infinitely different ways with respect to each other without violating any constraints imposed by $G$. For this reason we restrict our attention to the case where $G$ is connected. 

In this work, we are interested in a scalable system of $n$ unknown nodes for a large value of $n$.
As $n$ grows, it is reasonable to assume that the average number of connected neighbours 
for each node should stay constant. 
This happens, in our model, if we chose the radio range $R=C/n^{1/d}$. 
\begin{comment}
However, the number of hops is well defined only if the graph $G$ is connected.
If $G$ is not connected there might be a set of unknown nodes that are connected 
to too few anchors, resulting in under-determined series of equations in the triangulation step.
%there are multiple configurations resulting 
%in the same observed proximity measures and global localization is not possible.
\end{comment}
However, in the unit square, assuming sensor positions are drawn uniformly, 
the random geometric graph is connected, with high probability, 
if $\pi R^2 > (\log n +c_n)/ n$ for $c_n\rightarrow \infty$ \cite{GuK98}.
A similar condition can be derived for generic $d$-dimensions as 
$C_d R^d> (\log n +c_n)/ n$, where $C_d\leq\pi$ is a constant that depends on $d$. Moreover, in case $C_d R^d< (\log n +c_n)/ n$, not only the graph is not connected, there will be \textit{isolated nodes} with hight probability. Since isolated nodes cannot communicate with other sensors, there is no way to find their shortest paths to other nodes. Consequently, both {\sc MDS-MAP} and {\sc HOP-TERRAIN}  algorithms will be in trouble (see Figure~\ref{fig:scale}).
\begin{figure}[t]
\begin{center}
%\hspace{-2.1cm}
\includegraphics[width=7cm]{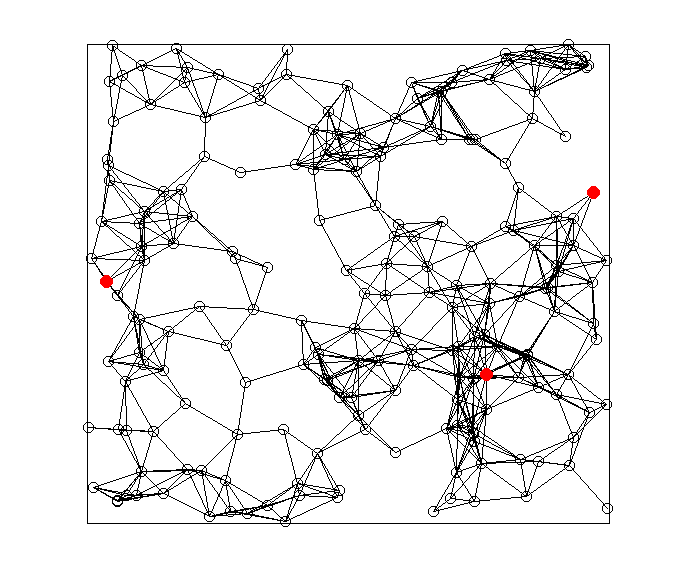}
\includegraphics[width=7cm]{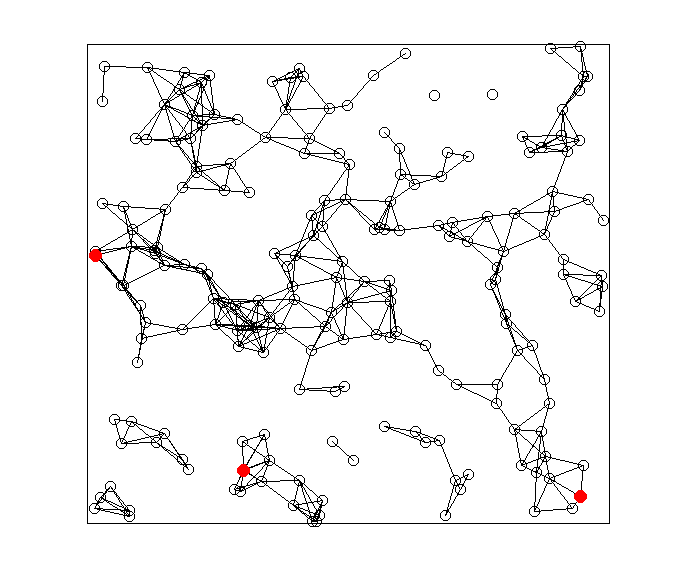}
\caption{The red vertices indicate the anchors. Under the right scaling of the radio range $R$, the graph stays connected (left figure) whereas otherwise there will be nodes without any means of communication to others (right graph).} \label{fig:scale}
\end{center}
\end{figure}
Hence, instead of $R=C/n^{1/d}$, 
we focus in the regime where the average number of connected neighbors is slowly increasing with $n$. 
Let $R_{\rm critical}$ be the critical detection range where the resulting graph starts to be connected. 
Then we are interested in the regime $R=CR_{\rm Critical}$,  
for some positive constant $C\geq1$ such that the graph stays connected with high probability. 

\begin{comment}
Note that $\hD$ is well defined only if the graph $G$ is connected.
Otherwise, there are multiple configurations resulting 
in the same observed proximity measures and global localization is not possible.
In the unit square, assuming sensor positions are drawn uniformly 
at random as define in the previous section, 
the graph is connected, with high probability, 
if $\pi R^2 > (\log n +c_n)/ n$ for $c_n\rightarrow \infty$ \cite{GuK98}.
A similar condition can be derived for generic $d$-dimensions as 
$C_d R^d> (\log n +c_n)/ n$, where $C_d\leq\pi$ is a constant that depends on $d$.
Hence, we focus in the regime where $R=(\alpha\log n/n)^{1/d}$ 
for some positive constant $\alpha$.
\end{comment}
In our analysis, the key observation and the crux of the argument is to show that the shortest-path 
 estimate is guaranteed to be arbitrarily close
to the correct distance for large enough radio range $R$ and large enough $n$. Once we proved this, we can then show that  the MDS step (equivalently, lateration) finds almost correctly the relative (equivalently, global) position of the sensors. We demonstrate how the error in estimating the  Euclidean distance will be reflected on the position estimation.  The precise statements are given in Section~\ref{sec:main}. 

We have already discussed the complexity of MDS and lateration steps. To complete our discussion we need to evaluate the complexity of finding the shortest path. In the {\sc MDS-MAP} algorithm we require that  all-pairs shortest paths be found. This problem has 
an efficient algorithm whose complexity is $O(n^2\log n + n|E|)$ \cite{Joh77}. 
For $R=C(\log n/n)^{1/d}$ with constant $C$, 
the graph is sparse with $|E|=O(n\log n)$, 
whence the complexity is $O(n^2\log n)$. Contrary to {\sc MDS-MAP}, in {\sc HOP-TERRAIN} we must only compute the shortest paths between the unknown nodes and the anchors. This distributed shortest paths algorithm can be done efficiently
with total complexity of $O(n\,m)$.
%==========================================================================%
\section{Main results}
\label{sec:main}

In this section we present our main results regarding the performance of {\sc MDS-MAP} and {\sc HOP-TERRAIN} algorithms. 
\subsection{{\sc MDS-MAP}}
Our first result establishes an upper bound on the error achieved by {\sc MDS-MAP} 
when we have only the connectivity information 
as in the case of the {\em connectivity-based model}. 

Let $\hX$ denote an $n \times d$ estimation for $X$ 
with an estimated position for node $i$ in the $i$th row. 
Then, we need to define a metric for the distance 
between the original position matrix $X$ 
and the estimation $\hX$, which is invariant under 
rigid transformation of $X$ or $\hX$. 

Define $\Ln \equiv \id_n-(1/n)\ones_n\ones_n^T$ as in the MDS algorithm.
$\Ln$ is an $n \times n$ rank $n-1$ symmetric matrix, 
which eliminates the contributions of the translation, 
in the sense that $LX=L(X+\ones s^T)$ for all $s\in R^d$. 
Note that $\Ln$ has the following nice properties:
\begin{enumerate}
\item  $\Ln XX^T\Ln$ { is invariant under rigid transformation}.  
\item $\Ln XX^T\Ln=\Ln \hX\hX^T\Ln$ { implies that $X$ and $\hX$ are equal up to a rigid transformation}. 
\end{enumerate} 
%(1) $\Ln XX^T\Ln$ { is invariant under rigid transformation}.  
%(2) $\Ln XX^T\Ln=\Ln \hX\hX^T\Ln$ { implies that $X$ and $\hX$ are equal up to a rigid transformation}. 
This naturally defines the following distance between $X$ and $\hX$.
\begin{eqnarray}
	\d(X,\hX) = \frac{1}{n}\big\|\Ln XX^T\Ln-\Ln\hX\hX^T\Ln \big\|_F \;, \label{eq:d1}
\end{eqnarray}
where $\|A\|_F = (\sum_{i,j}A_{ij}^2)^{1/2}$ denotes the Frobenius norm.
Notice that the factor $(1/n)$ corresponds to 
the usual normalization by the number of entries in the summation. 
Indeed this distance is invariant to rigid transformation of $X$ and $\hX$. 
Furthermore, $\d(X,\hX)=0$ implies that $X$ and $\hX$ 
are equal up to a rigid transformation.
With this metric, our
main result establishes an upper bound on the resulting error. 
The proof of this theorem is provided in Section \ref{sec:mainproof}.
We define 
\begin{eqnarray}
	R_{\text{MDS}} \equiv 32\left(\frac{12 \log n}{\alpha(n-2)}\right)^{\frac{1}{d}}  \;. \label{eq:defR0}
\end{eqnarray}

\begin{thm}[connectivity-based model]
\label{thm:main1}
Assume $n$ nodes are distributed uniformly at random in the $[0,1]^d$ hypercube, 
for a bounded dimension $d\in\{2,3\}$. 
For a positive radio range $R$ and detection probability 
$p$ defined in \eqref{eq:detection}, we are given 
the connectivity information of the nodes according to the range-free model with probabilistic detection.
Then, with a probability larger than $1-1/n^4$, the distance between the estimate $\hX$ produced by  
{\sc MDS-MAP} and the correct position matrix $X$ is bounded by
\begin{eqnarray}
	\d(X,\hX) \leq \frac{R_{\text{MDS}}}{R}\, + 20R \;, \label{eq:thm1} 
\end{eqnarray}
for $R>(1/\alpha)^{1/d}R_{\text{MDS}}$, where $\d(\cdot)$ is defined in \eqref{eq:d1} and $R_{\text{MDS}}$ in \eqref{eq:defR0}.
\end{thm} 
The proof is provided in Section \ref{sec:mainproof}. The following corollary trivially follows, as for each $(i,j)\in E$, we have $d_{i,j}\leq R$.
\begin{coro}[range-based model]
Under the hypotheses of Theorem~\ref{thm:main1} and in the case of rang-based model, with high probability $$\d(X,\hX) \leq \frac{R_{\text{MDS}}}{R}\, + 20R \;.$$
\end{coro}

As described in the previous section, 
we are interested in the regime where 
$R=C(\log n/n)^{1/d}$ for some constant $C$.
Given a small positive constant $\delta$, 
this implies that {\sc MDS-MAP} is guaranteed to produce estimated positions
that satisfy $\d(X,\hX)\leq \delta$ 
with a large enough constant $C$ and a large enough $n$.

When $\alpha$ is fixed and $R = C(\log n/n)^{1/d}$ for some positive parameter $C$, 
the error bound in \eqref{eq:thm1} becomes 
$$\d(X,\hX) \leq \frac{C_1}{C\alpha^{1/d}}+C_2C \left(\frac{\log n}{n}\right)^{1/d}, $$ 
for some numerical constants $C_1$ and $C_2$. 
The first term is inversely proportional to $C$ and $\alpha^{1/d}$ and is independent of $n$, 
whereas the second term is linearly dependent on $C$ and vanishes as
$n$ grows large. 
This is illustrated in Figure \ref{fig:graph_centralized}, which shows numerical simulations with $n$ sensors
randomly distributed in the 2-dimensional unit square.
Notice that the resulting error is inversely proportional to $\alpha$ and independent of $\beta$.

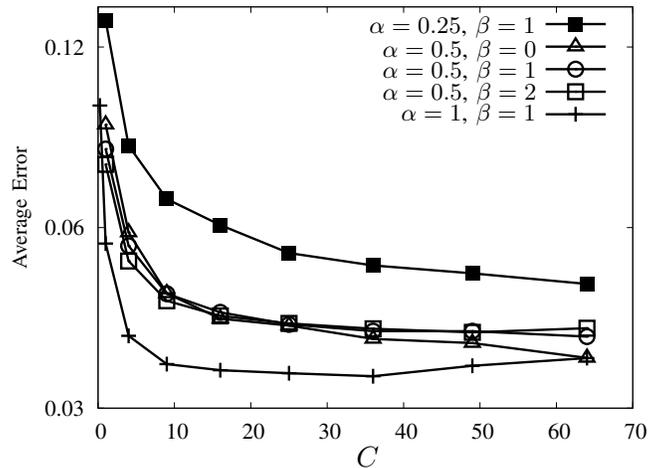
\begin{figure}[t]
\begin{center}
\vspace{-1cm}
\includegraphics[width=9cm]{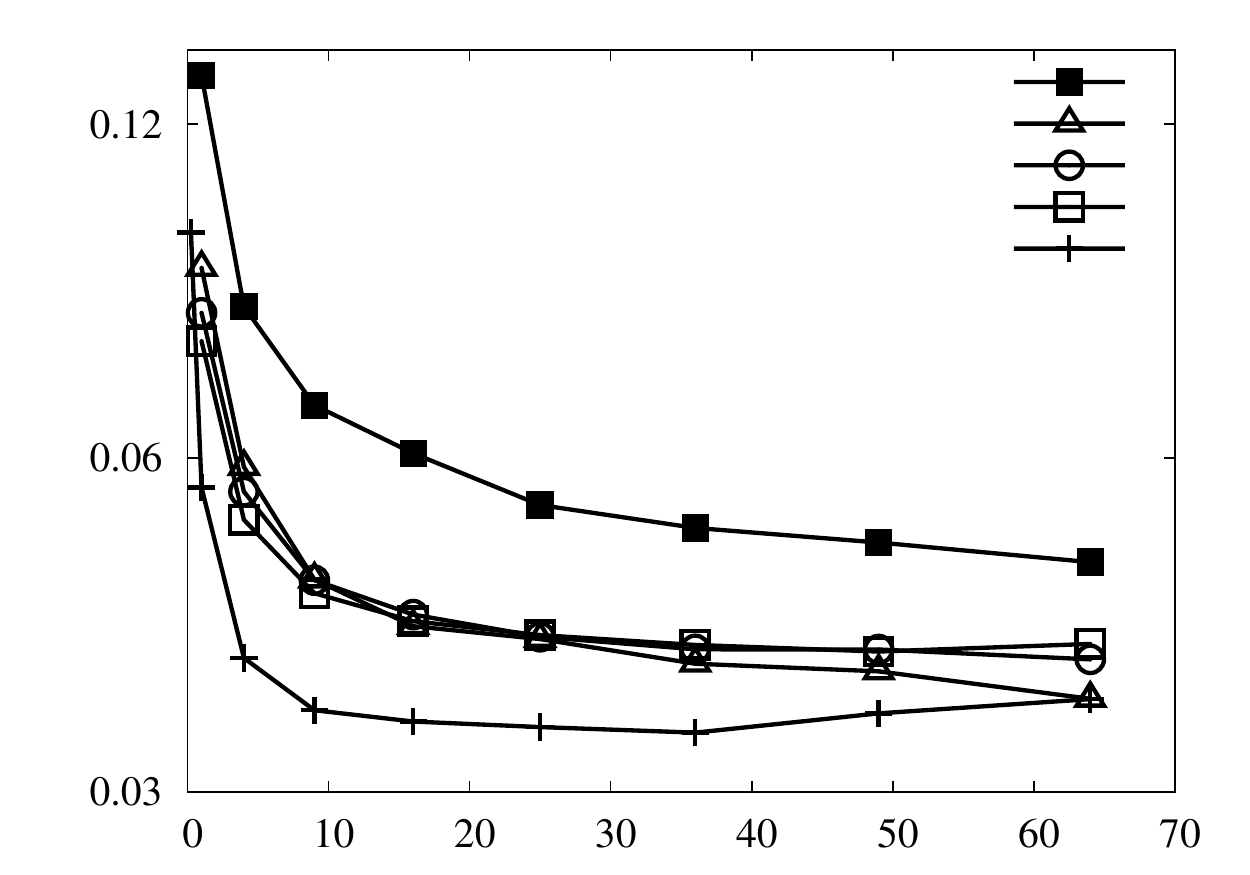}
\put(-120,-5){$C$}
\put(-250,70){\begin{sideways}\scriptsize{Average Error}\end{sideways}}
\put(-115,159){\footnotesize $\alpha=0.25$, $\beta=1$}
\put(-110,150){\footnotesize $\alpha=0.5$, $\beta=0$}
\put(-110,142){\footnotesize $\alpha=0.5$, $\beta=1$}
\put(-110,134){\footnotesize $\alpha=0.5$, $\beta=2$}
\put(-103,126){\footnotesize $\alpha=1$, $\beta=1$}
\caption{Average distance between the correct topology $X$ and the estimation $\hat{X}$ using 
{\sc MDS-MAP} as a function of $C$ where the radio range is $R=C\sqrt{\log n/n}$. 
The $n=1,000$ sensors are distributed randomly on a unit square under range-free model. 
Various values of $\alpha$ and $\beta$ are used where 
two nodes at distance $r$ are detected with probability $p(r)=\min\{1,\alpha(R/r)^\beta\}$. } 
\label{fig:graph_centralized}
\end{center}
\end{figure}
\begin{remark}
Even though the upper bounds for both range-free and range-based models have the same form, 
there is a slight difference between their behaviours as $R$ grows. In the range-free case, 
up to some point, the performance of {\sc MDS-MAP} improves as $R$ increases. 
This is due to the fact that the first and  second terms go in opposite 
directions as a function of $R$. However, In the range-based case, as $R$ 
increases, we obtain a  more accurate estimate of the the Euclidean distance. 
As a result, once the radio range increases, the resulting error of {\sc MDS-MAP} 
decreases and we do not see the contribution of the second term. 
This phenomenon is illustrated in Figure~\ref{fig:range-based}.
\end{remark}
\begin{figure}[t]
\begin{center}
%\vspace{-1cm}
\includegraphics[width=9cm]{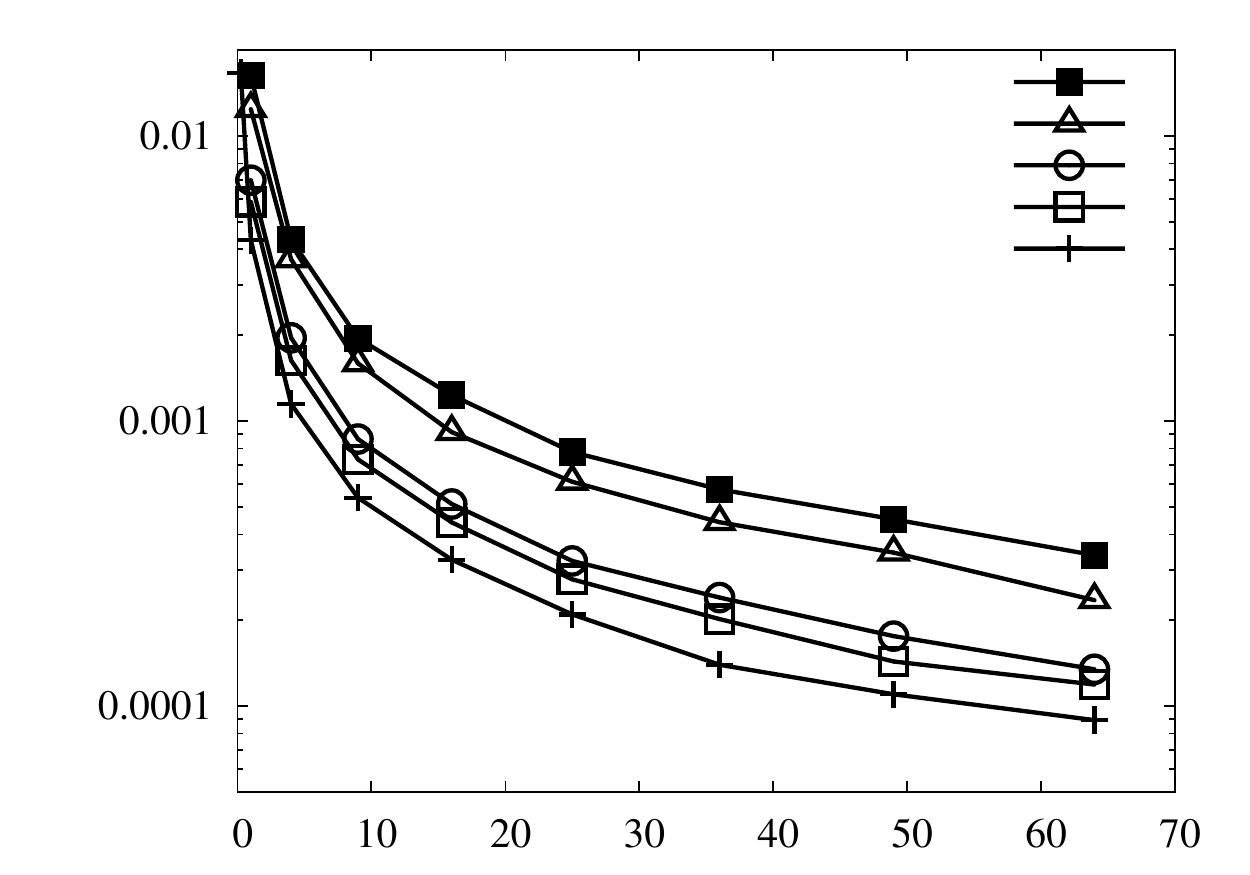}
\put(-120,-5){$C$}
\put(-250,70){\begin{sideways}\scriptsize{Average Error}\end{sideways}}
\put(-115,159){\footnotesize $\alpha=0.25$, $\beta=1$}
\put(-110,150){\footnotesize $\alpha=0.5$, $\beta=0$}
\put(-110,142){\footnotesize $\alpha=0.5$, $\beta=1$}
\put(-110,134){\footnotesize $\alpha=0.5$, $\beta=2$}
\put(-103,126){\footnotesize $\alpha=1$, $\beta=1$}
\caption{Average error of {\sc MDS-MAP}  under the range-based model. } 
%The error decreases as the radio range increases which shows a different behaviour than that of {\sc MDS-MAP}  under the range-free model } 
\label{fig:range-based}
\end{center}
\end{figure}
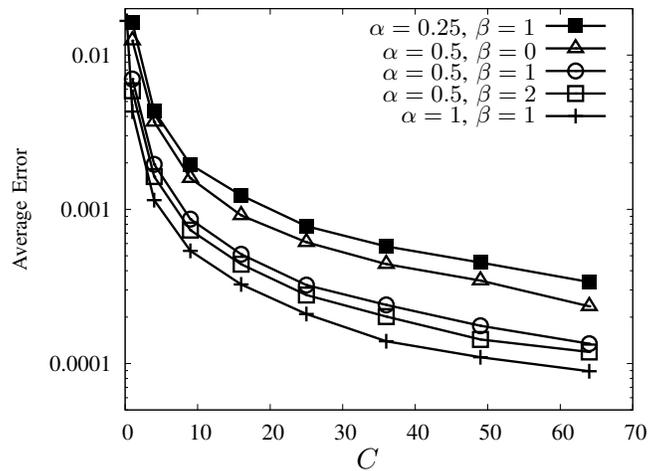
Using the above theorem, we can further show that there is a linear transformation $S\in \R^{d\times d}$, such that when applied to the estimations, we get a similar bound in the Frobenius norm of the error in the positions.
\begin{thm}\label{th:transform}
Under the hypotheses of Theorem~\ref{thm:main1}, with high probability $$ \min_{S\in \R^{d\times d}} \frac{1}{\sqrt{n}}\| LX-L\hX S\|\leq \sqrt{6}\left(\frac{R_{\text{MDS}}}{R}+20R\right)$$
\end{thm}
\begin{remark}
Note that although for the sake of simplicity, we focus on $[0,1]^d$ hypercube;
our analysis easily generalizes to any bounded convex set
and homogeneous Poisson process model with density $\rho=n$.
The homogeneous Poisson process model is characterized by the probability 
that there are exactly $k$ nodes appearing in any region with volume $A$ :
$\prob(k_A=k)= \frac{(\rho A)^k}{k!}e^{-\rho A}$.
Here, $k_A$ is a random variable defined as 
the number of nodes in a region of volume $A$. 
\end{remark}
\begin{remark}
To simplify calculations, we assumed that $d$ is either $2$ or $3$. 
However, the analysis easily applies to general $d$ and 
only the constant in the bound \eqref{eq:thm1} would change as long as $d=O(1)$. 
\end{remark}

In what follows we investigate an important question
whether similar performance guarantees, as in {\sc MDS-MAP}, can be obtained
in a distributed setting. In particular, we analyze the performance
of the {\sc HOP-TERRAIN} algorithm. As we have already stressed, this algorithm can be seen as a distributed
version of the {\sc MDS-MAP} algorithm. In particular, we show that
when only connectivity information is available, for every
unknown node. The Euclidean distance between the estimate
 and the correct position can be bounded very  similarly to Theorem~\ref{thm:main1}.
 %============================================================================
\subsection{{\sc HOP-TERRAIN}}
Our second result establishes that {\sc HOP-TERRAIN} \cite{SLR02} achieves an arbitrarily small error 
for a radio range $R=C(\log n/n)^{1/d}$ with a large enough constant $C$, 
when we have only the connectivity information 
as in the case of the {\em connectivity-based model}. 
The same bound holds immediately for the {\em range-based model}, 
when we have an approximate measurements for the distances, and 
the same algorithm can be applied without any modification. 
to compute better estimates for the actual distances between the unknown nodes and the anchors, the extra information can be readily incorporated into the algorithm. We define 
\begin{eqnarray}
R_{\text{HOP}} \equiv 12\left(\frac{12 \log n}{\alpha(n-2)}\right)^{\frac{1}{d}}  \;. \label{eq:rHOP}
%	R_1 \equiv {C\,d^{2}}\left(\frac{\log n}{n}\right)^{\frac{1}{d}}  \;. \label{eq:r1}
\end{eqnarray}

\begin{thm}
Assume $n$ sensors and $m$ anchors are distributed 
uniformly at random in the $[0,1]^d$ hypercube 
for a bounded dimension $d\in \{2,3\}$.
For a given radio range $R>(1/\alpha)^{1/d}R_{\text{HOP}}$, detection probability $p$ defined in \eqref{eq:detection}, and the number of anchors $m=\Omega(\log n)$, 
the following is true with  probability  at least $1-1/n^4$. 
For all unknown nodes $i\in V_u$, 
the Euclidean distance between the estimate $\hx_i$ given by 
{\sc HOP-TERRAIN} and the correct position $x_i$ is bounded by
\begin{eqnarray}
	\|x_i-\hx_i\| \leq \frac{R_{\text{HOP}}}{R}\, + 24R\;.
\end{eqnarray}
\label{thm:main2}
\end{thm} 
The proof is provided in Section \ref{sec:mainproof}.
% implication of thm 1.1 : 
%In the regime where $R=o(1)$, the above theorem implies that 
%the error is inversely proportional to the radio range $R$.
As described in the previous section, 
we are interested in the regime where 
$R=C(\log n/n)^{1/d}$ for some constant $C$.
Given a small positive constant $\delta$, 
this implies that {\sc HOP-TERRAIN} is guaranteed to produce estimated positions
that satisfy $\|x_i-\hx_i\| \leq \delta$ for all $i\in V_u$
with a large enough constant $\alpha$ and large enough $n$.

When the number of anchors is bounded and 
the positions of the anchors are chosen randomly, 
it is possible that, in the triangulation step,
we get an ill-conditioned matrix $A^TA$, 
resulting in an large estimation error.
This happens, for instance, if three anchors fall close to a line.
However, as mentioned in the introduction, 
it is reasonable to assume that, for the anchors, 
the system designer has some control over where they are placed.
In that case, the next remark shows that 
when the positions of anchors are properly chosen, 
only $d+1$ anchors suffice 
to get a similar bound on the performance.
Note that this is the minimum number of anchors necessary for triangulation.
For simplicity we assume that one anchor is placed at the origin 
and $d$ anchors are placed at positions corresponding to 
$d$-dimensional unit vectors. The position of the $d+1$ anchors are 
\{$[0,\ldots,0]$, $[1,0,\ldots,0]$, $[0,1,0,\ldots,0]$, $[0,\ldots,0,1]$ \}. (see figure \ref{fig:RGG_det})

\begin{figure}[t]
\begin{center}
%\vspace{-1cm}
\includegraphics[width=9cm]{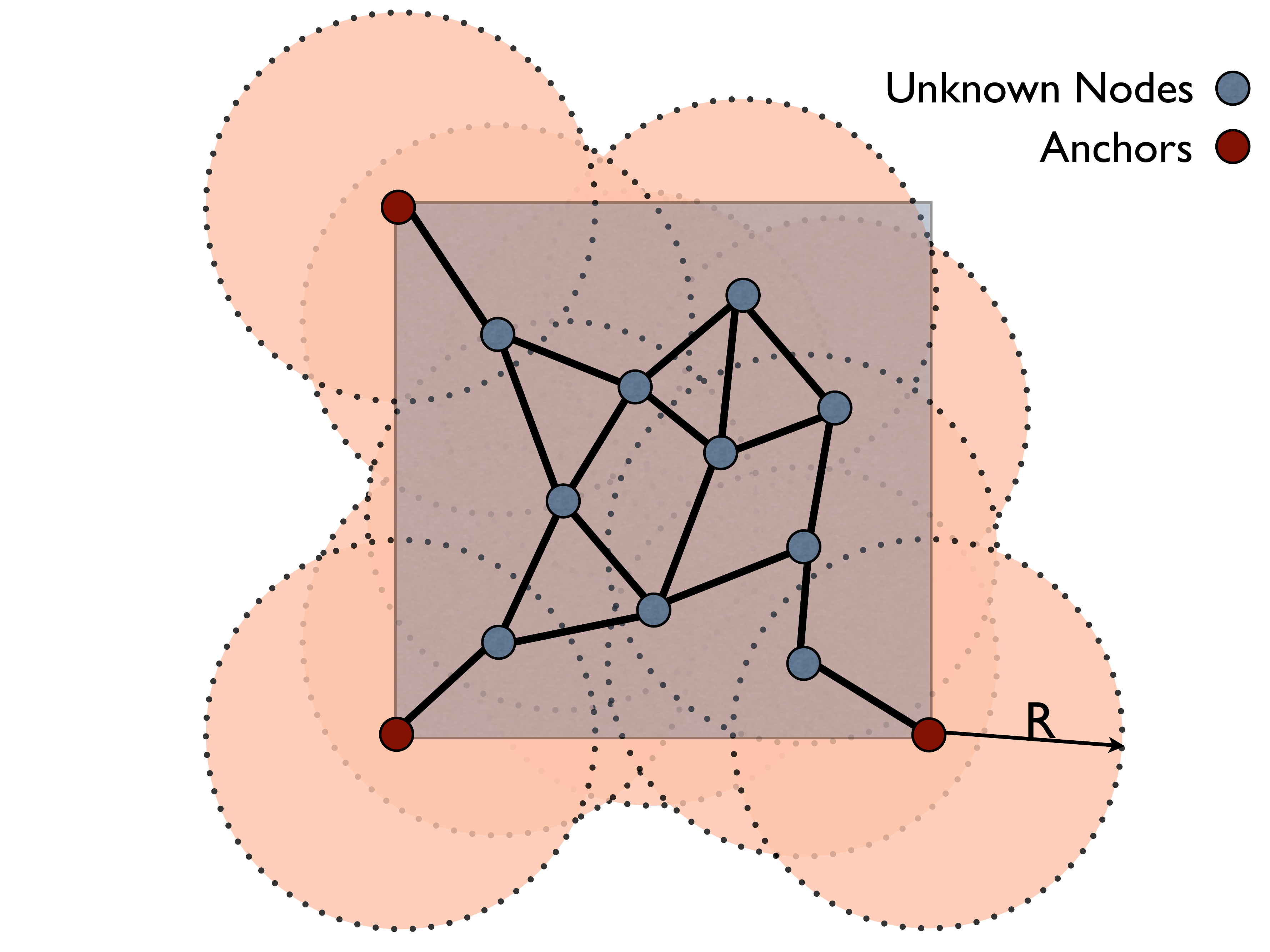}
\caption{ Three anchors in fixed positions ($[0,0], [1,0], [0,1]$) for a two-dimensional sensor localization.} \label{fig:RGG_det}
\end{center}
\end{figure}
\begin{thm}
\label{thm:main3}
Assume that $n$ sensors are distributed 
uniformly at random in the $[0,1]^d$ hypercube 
for a bounded dimension $d=\{2,3\}$. Also, assume that
there are $d+1$ anchors, one of which is placed at the origin, 
and the position vectors of the $d$ remaining anchors are the $d$-dimensional unit vectors.
For a given radio range $R>(1/\alpha)^{1/d}R_{\text{HOP}}$  and detection probability $p$ defined in \eqref{eq:detection}
the following is true with 	probability at least $1-1/n^4$. 
For all unknown nodes $i\in V_u$, 
the Euclidean distance between the estimate $\hx_i$ given by 
\hopterrain ~and the correct position $x_i$ is bounded by
\begin{eqnarray}
	\|x_i-\hx_i\| \leq 2 \frac{R_{\text{HOP}}}{R}\, + 48 R \;. \label{eq:main1}
\end{eqnarray}
\end{thm}
The proof is provided in Section \ref{sec:mainproof}.
\begin{remark}
There is nothing particular about the position of the anchors 
in unit vectors. Any $d+1$ anchors in general position will give similar bounds.
The only difference is that the constant term in the definition of $R_{\text{HOP}}$ 
changes with the anchor positions. 
\end{remark}
\begin{figure}[t]
\begin{center}
%\hspace{-2.1cm}
\includegraphics[width=9cm]{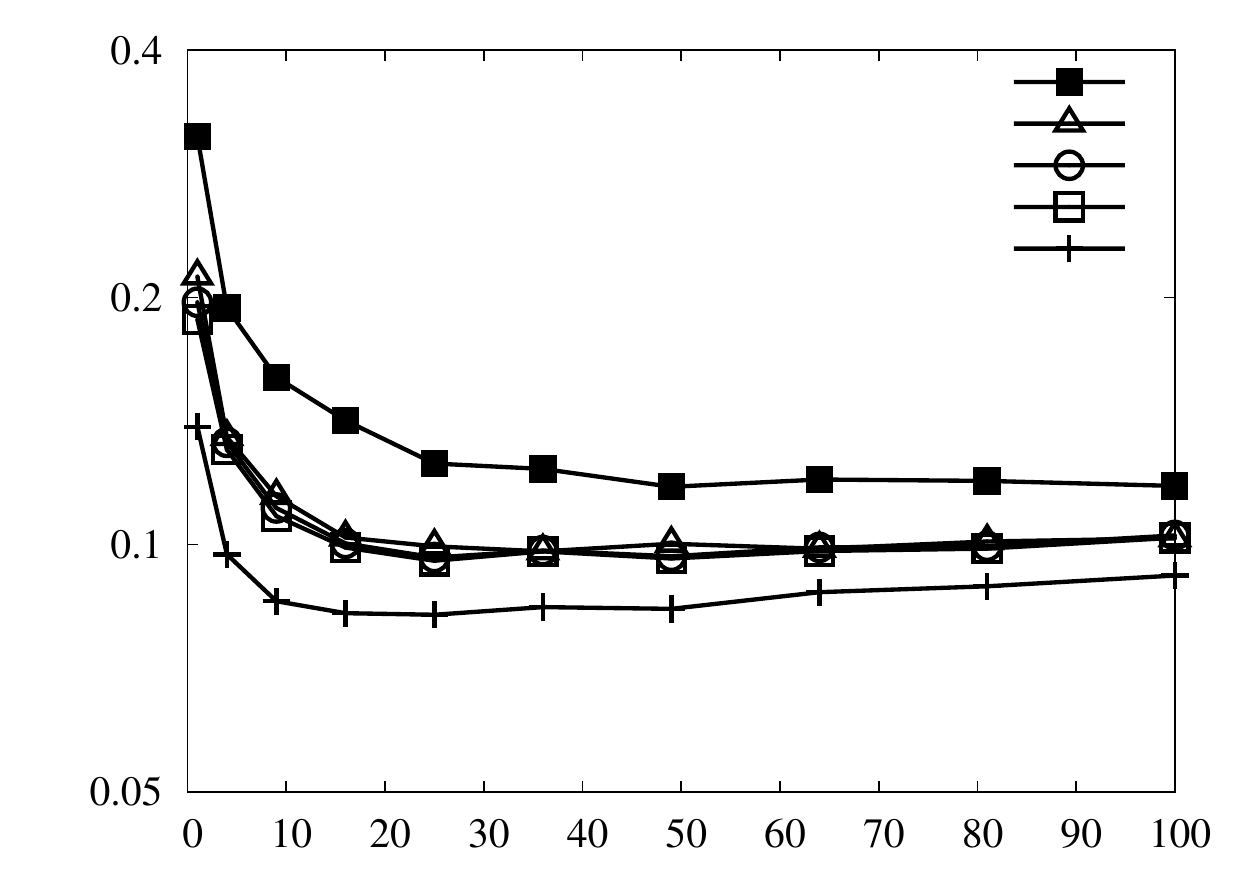}
%\put(-102,-5){$C$}
%\put(-223,70){\begin{sideways}\scriptsize{Average Error}\end{sideways}}
\put(-120,-5){$C$}
\put(-250,70){\begin{sideways}\scriptsize{Average Error}\end{sideways}}
\put(-115,159){\footnotesize $\alpha=0.25$, $\beta=1$}
\put(-110,150){\footnotesize $\alpha=0.5$, $\beta=0$}
\put(-110,142){\footnotesize $\alpha=0.5$, $\beta=1$}
\put(-110,134){\footnotesize $\alpha=0.5$, $\beta=2$}
\put(-103,126){\footnotesize $\alpha=1$, $\beta=1$}
%\includegraphics[width=8cm]{RMSEdist_RB.eps}
%\put(-102,-5){$\alpha$}
%\put(-223,70){\begin{sideways}\scriptsize{Average Error}\end{sideways}}
%\put(-345,90){\begin{sideways}\scriptsize{Relative Error}\end{sideways}}
\caption{Average distance between the correct position $\{x_i\}$ and estimation $\{\hx_i\}$ 
	using {\sc Hop-TERRAIN} as a function of $C$, 
	for $R=C\sqrt{\log n/n}$ with $n=5,000$ sensors in the unit square 
        under connectivity-based model. 
	Two nodes at distance $r$ detect each other with probability $p(r)=\min\{1,\alpha(R/r)^\beta\}$.} 
\label{fig:decentral_connectivity}
\end{center}
\end{figure}
\begin{figure}[t]
\begin{center}
%\hspace{-2.1cm}
\includegraphics[width=9cm]{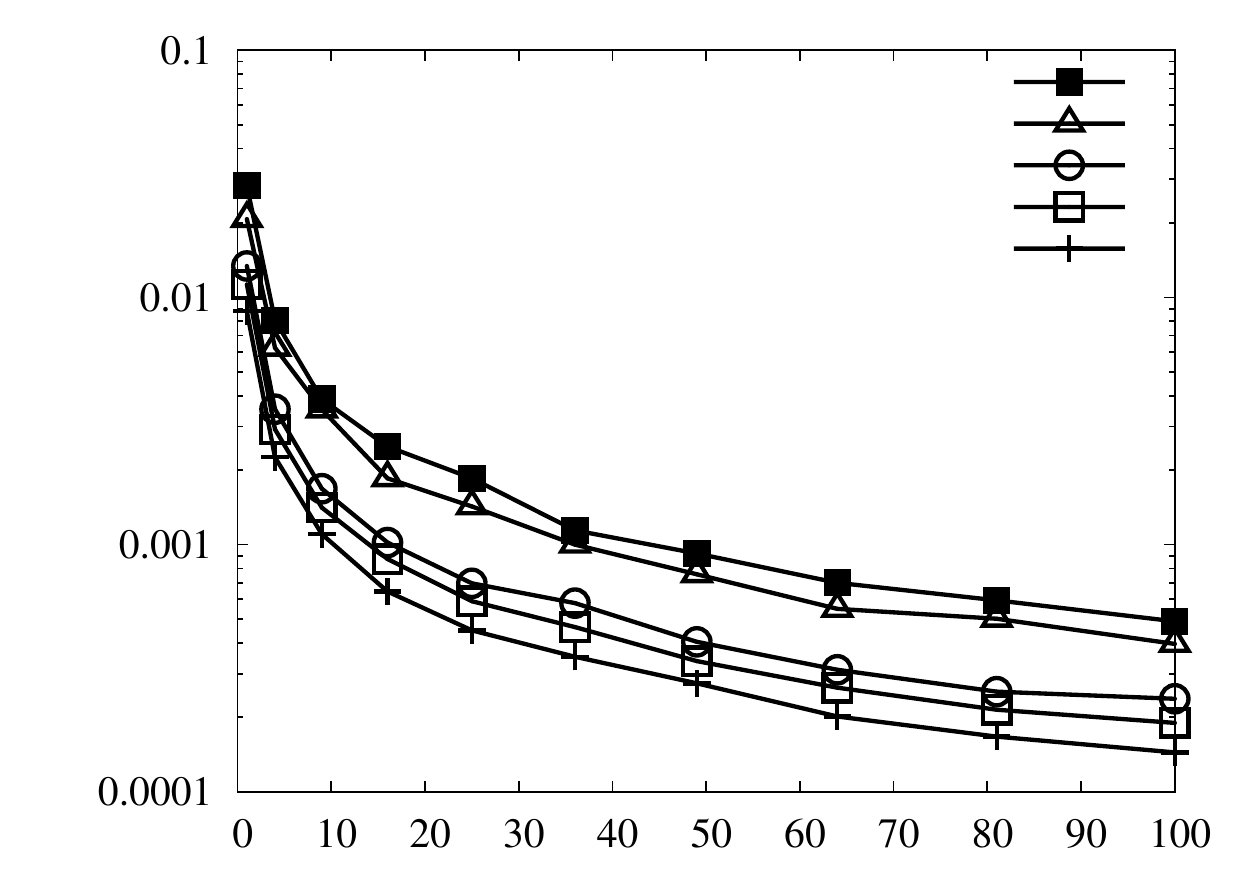}
%\put(-102,-5){$C$}
%\put(-223,70){\begin{sideways}\scriptsize{Average Error}\end{sideways}}
\put(-120,-5){$C$}
\put(-250,70){\begin{sideways}\scriptsize{Average Error}\end{sideways}}
\put(-115,159){\footnotesize $\alpha=0.25$, $\beta=1$}
\put(-110,150){\footnotesize $\alpha=0.5$, $\beta=0$}
\put(-110,142){\footnotesize $\alpha=0.5$, $\beta=1$}
\put(-110,134){\footnotesize $\alpha=0.5$, $\beta=2$}
\put(-103,126){\footnotesize $\alpha=1$, $\beta=1$}
%\includegraphics[width=8cm]{RMSEdist_RB.eps}
%\put(-102,-5){$\alpha$}
%\put(-223,70){\begin{sideways}\scriptsize{Average Error}\end{sideways}}
%\put(-345,90){\begin{sideways}\scriptsize{Relative Error}\end{sideways}}
\caption{Average error under range-based model. } 
\label{fig:decentral_distance}
\end{center}
\end{figure}
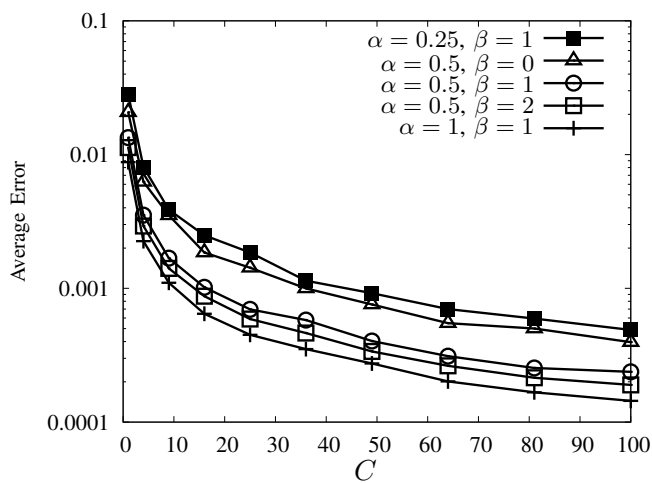

\begin{coro}[range-based model]
Under the hypothesis of Theorem~\ref{thm:main2} and in the range-based model, with high probability
\begin{eqnarray*}
	\|x_i-\hx_i\| \leq \frac{R_{\text{HOP}}}{R}\, + 24R \;.
\end{eqnarray*}
The similar result holds true when sensors are places deterministically, specifically, under the hypothesis of Theorem~\ref{thm:main3}, with high probability,
\begin{eqnarray*}
	\|x_i-\hx_i\| \leq 2 \frac{R_{\text{HOP}}}{R}\, + 48 R \;.
\end{eqnarray*}

\label{coro:main3}
\end{coro} 
As it was the case for {\sc MDS-MAP}, when $R=C(\log n / n)^{1/d}$ for some positive parameter $C$,  
the error bound in (\ref{eq:main1}) is 
$$\|x_i-\hx_i\| \leq \frac{C_1}{C\alpha^{1/d}} + C_2C\left(\frac{\log n}{n}\right)^{1/d}$$ 
for some numerical constants $C_1$ and $C_2$.
The first term is inversely proportional to $C$ and $\alpha^{1/d}$ and is independent of $n$,
whereas the second term is linearly dependent in $C$ and vanishes as $n$ grows large.
This is illustrated in Figure \ref{fig:decentral_connectivity}, 
which shows numerical simulations with $n=5,000$ sensors randomly distributed 
in the $2$-dimensional unit square. 
We compute the root mean squared error: $\{(1/n)\sum_{i=1}^{n}\|x_i-\hx_i\|^2\}^{1/2}$.

Figure~\ref{fig:hopterrain1} shows a network consisting of $n=200$ nodes place randomly in the unit circle. The three anchors in fixed positions are displayed by solid blue circles. In this experiment the distance measurements are from the range-based model and the radio range is $\sqrt{0.8 \log n/n}$. Figure~\ref{fig:hopterrain2}  shows the final estimated positions using {\sc HOP-TERRAIN}. The circles represent the correct positions, and the solid lines represent the differences between the estimates and the correct positions. The average error in this example is $0.075$. 

\begin{figure}[h]
\begin{center}
\includegraphics[width=9cm]{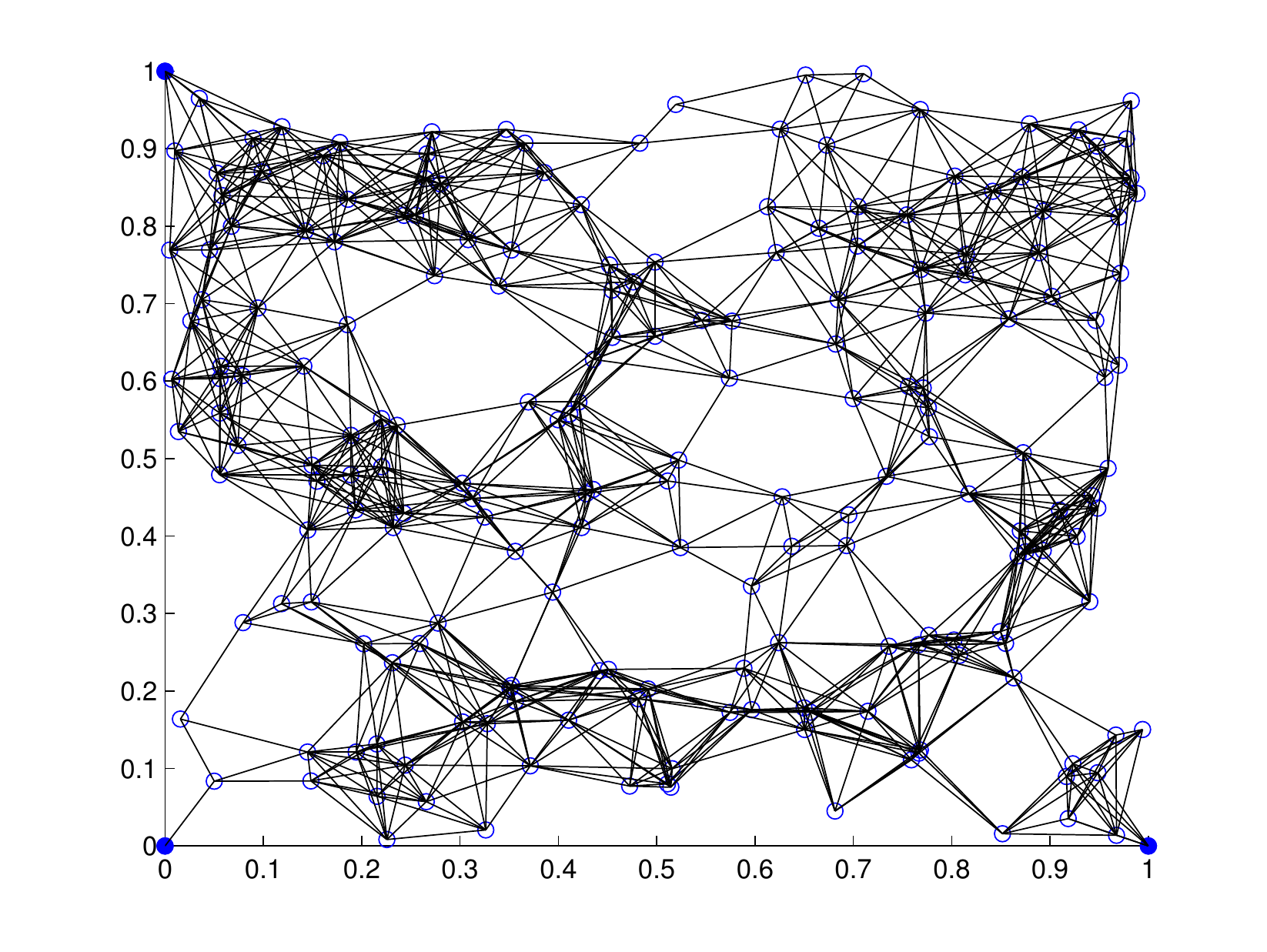}
\caption{ $200$ nodes randomly placed in the unit square and $3$ anchors in fixed positions. The radio range is $R=\sqrt{0.8*\log n/n}$.} \label{fig:hopterrain1}
\end{center}
\end{figure}
\begin{figure}[h]
\begin{center}
\includegraphics[width=9cm]{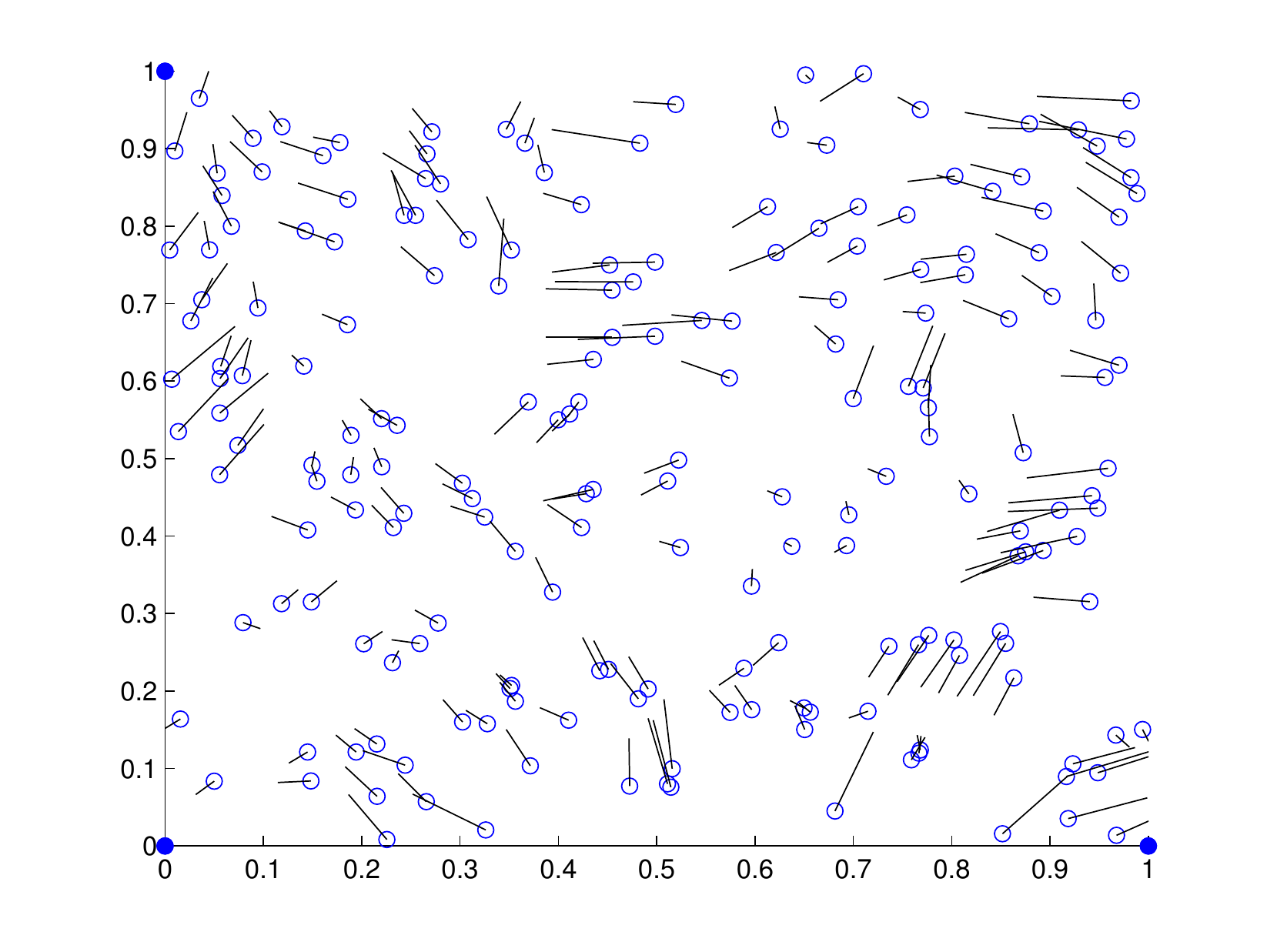}
\caption{Location estimation using {\sc Hop-TERRAIN}.} \label{fig:hopterrain2}
\end{center}
\end{figure}

\section{Proof of the Main Theorems} 
\label{sec:mainproof}

\subsection{Proof of Theorem \ref{thm:main1}}

We start by bounding the distance $\d(X,\hX)$, as defined in Eq.~(\ref{eq:d1}),
in terms of $D$ and $\hD$. 
Let $\|A\|_F=(\sum_{i,j}A_{ij}^2)^{1/2}$ denote the Frobenius norm of a matrix and 
$\|A\|_2=\max_{\|x\|=1} \|Ax\|_2$ denote the spectral norm. Note that for a rank $r$ matrix $A$ we have $$\|A\|_2\leq\|A\|_F\leq r\|A\|_2.$$
Since $\Ln (XX^T-\hX\hX^T)\Ln$ has rank at most $2d$, 
it follows that
\begin{align}
\|\Ln(XX^T-\hX\hX^T)\Ln\|_F  \leq \sqrt{2d}\|\Ln(XX^T-\hX\hX^T)\Ln\|_2 \;. \label{eq:distance}
\end{align}
%
%where we used the fact that, for any rank $2d$ matrix $A$, $\|A\|_F \leq \sqrt{2d}\|A\|_2$.
To bound the spectral norm, let $M=-(1/2)\Ln\hD\Ln$. Then, 
\begin{eqnarray}
 \|\Ln(XX^T-\hX\hX^T)\Ln\|_2 &\leq& \|\Ln XX^T\Ln-M\|_2 + \|M-\hX\hX^T \|_2 \nonumber\\
 &\leq& (1/2)\|\Ln(-D+\hD)\Ln\|_2 + (1/2)\|\Ln(-\hD+D)\Ln\|_2 \nonumber\\
 &\leq& \|\hD-D\|_2 \;, \label{eq:spectralnorm}
\end{eqnarray}
where in the first inequality we used the triangular inequality 
and the fact that $\hX=\Ln\hX$.
In the second inequality we used \eqref{eq:LDL} and 
the fact that 
\begin{align*} 
 \|M-\hX\hX^T\|_2 =\min_{A:{\rm rank}(A)\leq d} \|M-A\|_2\;,
\end{align*} 
which follows from the definition of $\hX$. 
From the definition of $\hX=\MDS_d(\hD)$, 
we know that $\hX\hX^T$ is the best rank-$d$ approximation to $M$.
Hence, $\hX\hX^T$ minimizes $\|M-A\|_2$ for any rank-$d$ matrix $A$.
Since the rank of $-(1/2)\Ln D\Ln$ is $d$, this implies  $$\|M-\hX\hX^T\|_2 \leq \|M+(1/2)\Ln D\Ln\|_2.$$
The inequality in (\ref{eq:spectralnorm}) follows trivially 
from the observation that $\|\Ln\|_2=1$.

Next, to bound $\|\hD-D\|_2$, we use the following key result on the number of hops in graph $G$. 
The main idea is that, for sensors with uniformly random positions, 
the number of hops scaled by the radio range $R$ provide 
estimates close to the correct distance. 
We define 
\begin{eqnarray}
	\tR \equiv 2 \left(\frac{12\log n}{\alpha (n-2)}\right)^{\frac{1}{d}} \;. \label{eq:tR}
\end{eqnarray}
\begin{lemma} 
	{\rm (Bound on the distance estimation)} 
	Under the hypotheses of Theorem \ref{thm:main1}, 
	with probability larger than $1-1/n^4$, 
	for any pair of nodes $i\in V$ and $j\in V$, 
	the number of hops between nodes $i$ and $j$ is bounded by
	\begin{eqnarray*}
		h_{i,j} \leq \Big(1+\frac{\tR}{R}\Big)\frac{d_{i,j}}{R} + 2 \; , 
	\end{eqnarray*}
	for $R>\max\{7\tR,(1/\alpha)^{1/d}\tR\}$.
	\label{lem:hops}
\end{lemma}
The proof of this lemma is provided in Section \ref{sec:hops}.
The distance estimate from the first step of {\sc MDS-MAP} is $\hd_{i,j}=R h_{i,j}$. The following corollary gives a bound on the estimation error.  
\begin{coro}
 	Under the hypotheses of Lemma \ref{lem:hops}, 
 	\begin{eqnarray*}
   	\hd_{i,j}^2 - d_{i,j}^2 \leq \frac{30\tR}{14R}d_{i,j}^2 + 8R\;.
 	\end{eqnarray*}
	\label{cor:hops}
\end{coro}

\begin{proof}
 	From Lemma \ref{lem:hops}, we know that 
 	\begin{eqnarray*}
   	(R\,h_{i,j})^2-d_{i,j}^2\leq \frac{2\tR}{R}\Big(1+\frac{\tR}{2R}\Big)d_{i,j}^2+2R\Big(1+\frac{\tR}{R}\Big)d_{i,j}+ 4R^2\;.
 	\end{eqnarray*}
 	The corollary follows from the assumption that $7\tR<R\leq1$ and $d\leq3$.
\end{proof}

Define an error matrix $Z=\hD-D$. Then by Corollary \ref{cor:hops},
$Z$ is element-wise bounded by $$0\leq Z_{ij} \leq (30\tR/(14R)) D_{ij} + 8R.$$
We can bound the spectral norm of $Z$ as follows. 
Let $u$ and $v$ be the left and right singular vectors of the non-negative matrix $Z$, respectively. 
Then by Perron-Frobenius theorem, $u$ and $v$ are also non-negative.
It follows that
\begin{eqnarray}
 	\|\hD-D\|_2 &=& u^TZv \nonumber\\
	 &\leq& (30\tR/(14R)) u^T Dv + (\ones^Tu)(\ones^Tv)8R \nonumber\\
	 &\leq& (30\tR/(14R)) \|D\|_2 + 8Rn \nonumber\\
	 &\leq& (30\tR/(14R))dn + 8Rn \label{eq:spectralnorm2} \;.
\end{eqnarray}
The first inequality follows from 
the element-wise bound on $Z$ and 
the non-negativity of $u$ and $v$, 
and the second inequality follows form the definition of the spectral norm and 
the Cauchy-Schwarz inequality. 
In the last inequality, we used $\|D\|_2\leq dn$, 
which follows from the fact that $D$ is non-negative and element-wise bounded by $d$. 
Typically we are interested in the regime where $R=o(1)$, 
and by assumption we know that $R\geq\tR$ and $d\leq 3$. 
Therefore, the first term in \eqref{eq:spectralnorm2} dominates the error.
Substituting this bound on $\|\hD-D\|_2$ in \eqref{eq:spectralnorm} proves the theorem.

%
%============================================================
%

\subsection{Proof of Theorem~\ref{th:transform}}
Using SVD we can write $LX$ as $U_{n\times d}\Sigma_{d \times d} V^T_{d\times d}$ where  $U^TU=\id_{d\times d}$, $V^TV=VV^T=\id_{d\times d}$ and $\Sigma$ is a diagonal matrix. We also denote the \textit{Frobenius inner product} between to matrices $A_{m\times n}$ and $B_{m\times n}$ by $$\langle A,B\rangle \doteq \sum_{i,j}A_{i,j} B_{i,j}.$$
It is easy to show that $$\langle A,B\rangle = \Tr(A^TB)\leq \|A\|_F \|B\|_F.$$

 In fact, this inner product induces the Frobenius norm definition. In particular, for an $m\times n$ matrix $A$ we have $$\|A\|_F= \sup_{B\in\R^{m\times n}, \|B\|_F\leq 1}\langle B,A\rangle.$$
Now, for $S=\hX^TLU\Sigma^{-1}V^T$, we have
\begin{eqnarray*}
\| LX-L\hX S\|_F &=&  \sup_{B\in\R^{n\times d}, \|B\|_F\leq 1}\langle B,LX-L\hX S\rangle \\
&=&  \sup_{B\in\R^{n\times d}, \|B\|_F\leq 1}\langle B,(LXV\Sigma U^T-L\hX\hX^TL)U\Sigma^{-1}V^T\rangle \\
&=&  \sup_{B\in\R^{n\times d}, \|B\|_F\leq 1}\langle BV\Sigma^{-1}U^T,LXX^TL-L\hX\hX^TL\rangle \\
&=&  \sup_{B\in\R^{n\times d}, \|B\|_F\leq 1}\| BV\Sigma^{-1}U^T\|_F\|LXX^TL-L\hX\hX^TL\|_F.
\end{eqnarray*}
Using the fact $\|A\|_F=\Tr(A^TA)$ and the cyclic property of the trace, i.e., $\Tr(ABC)=\Tr(BCA)$, we obtain $$\| BV\Sigma^{-1}U^T\|_F=\Tr( BV\Sigma^{-2}V^TB)\leq \sigma^2_{\min}\|B\|^2_F,$$
where $\sigma_{\min}$ is the smallest singular value of $LX$. It remains  to show that $\sigma_{\min}\geq \sqrt{n/6}$ holds with high probability when nodes are placed uniformly at random. To this end we need to consider two facts. First, the  singular values (and in particular  the smallest singular value) are Lipschitz functions of the entries (See appendix). Second, we have $E(LX\hX L)=(n/12)\id_{d\times d}$. By using concentration of measure  for Lipschitz functions on bounded independent random variables, the result follows.
%%%%%%%%%%%%%%%%%%%%%%%%%%%%%%%%%%%%%%%%%%%%%%%%%%%%%%%%%%%%%%%%%%%%%%%%%%
\subsection{Proof of Theorem~\ref{thm:main2}}
In this section we provide the proofs of the theorems \ref{thm:main2}. Detailed proofs of the technical lemmas are provided in the following sections.

For an unknown node $i$, 
the estimation $\hx_i$ is given in Eq.~(\ref{eq:lateration}). 
\begin{eqnarray}
 \|x_i-\hx_i\| &=&    \|(A^TA)^{-1}A^Tb_0^{(i)}-(A^TA)^{-1}A^Tb^{(i)}\| \nonumber\\
               &\leq& \|(A^TA)^{-1}A^T\|_2\|b_0^{(i)}-b^{(i)}\| \;, \label{first-step}
\end{eqnarray}
First, to bound $\|b_0^{(i)}-b^{(i)}\|$, we use Corollary~\ref{cor:hops}. Since $d_{i,j}^2\leq d$ for all $i$ and $j$, we have
\begin{eqnarray}
 \|b_0^{(i)}-b^{(i)}\| &=   & \Big(\sum_{k=1}^{m-1}\big(d_{i,k+1}^2-d_{i,k}^2-\hd_{i,k+1}^2+\hd_{i,k}^2 \big)^2 \Big)^{1/2}\nonumber\\
                       &\leq& 2\sqrt{m-1}\left(\frac{30\tR}{14R}d + 8R\right) \;,\label{eq:term1}
\end{eqnarray}
%\looseness -1
Next, to bound $\|(A^TA)^{-1}A^T\|_2$, we use the following lemma. 
\looseness -1
\begin{lemma} 
\label{lem:spectral-random}
Under the hypothesis of Theorem \ref{thm:main2}, 
the following is true. 
Assuming random anchor model  in which $m=\Omega(\log n)$ anchors are chosen uniformly at random among $n$ sensors. Then we have 
\begin{eqnarray*}
	\|(A^TA)^{-1}A^T\|_2 \leq \sqrt{\frac{3}{m-1}} \;, 
\end{eqnarray*}
with high probability.
\end{lemma}
 By assumption we know that $R\geq\tR$ and $d\leq 3$. By combining \eqref{first-step}, \eqref{eq:term1} and Lemma~\ref{lem:spectral-random} proves Theorems \ref{thm:main2}.
%%%%%%%%%%%%%%%%%%%%%%%%%%%%%%%%%%%%%%%%%%%%%%%%%%%%%%%%%%%%%%%%%%%%%%%%%%%%%%%%%%%%%%%%%%%%%
\subsection{Proof of  Theorem~\ref{thm:main3}}
In this section we provide the proof of Theorem \ref{thm:main3}.
Detailed proofs of the technical lemmas are provided in the following sections.

Similarly to  the proof of Theorem~\ref{thm:main2}, for an unknown node $i$, 
and the estimate  $\hx_i$ we have 
\begin{eqnarray*}
 \|x_i-\hx_i\| &\leq&    \|(A^TA)^{-1}A^T\|_2\|b_0^{(i)}-b^{(i)}\| \;,
\end{eqnarray*}
We have already bounded the expression $\|b_0^{(i)}-b^{(i)}\|$ in \eqref{eq:term1}.
To bound $\|(A^TA)^{-1}A^T\|_2$, we use the following lemma. 
\looseness -1
\begin{lemma} 
\label{lem:spectraldeterministic}
Under the hypothesis of Theorem \ref{thm:main3}, 
the following are true. 
We assume a deterministic anchor model, where $m=d+1$ anchors are 
placed on the positions %$x_1=[1,0,\ldots,0]$, $x_2=[0,1,0,\ldots,0]$, $x_3=[0,\ldots,0,1] $ and $x_m=[0,0,\ldots,0]$.
\begin{eqnarray*}
x_1&=&[1,0,\ldots,0],\\
x_2&=&[0,1,0,\ldots,0],\\
&\vdots & \\
x_d &=&[0,0,\ldots,0, 1],\\
x_{d+1}&=&[0,0,\ldots,0].
\end{eqnarray*}
Then,
\begin{eqnarray*}
	\|(A^TA)^{-1}A^T\|_2 \leq \frac{d}{2}  \;,
\end{eqnarray*}
with high probability. This finishes the proof of Theorems  \ref{thm:main3}.
\end{lemma}

\subsection{Proof of Lemmas \ref{lem:spectral-random} (Random Model) }
As it was the case in the proof of Lemma \ref{lem:spectraldeterministic} in order to upper bound $\|(A^TA)^{-1}A\|_2$ we need to lower bound the smallest singular value of $A$.   Let the symmetric matrix $B$ be defined as $A^TA$. The diagonal entries of $B$ can be written as
\begin{equation}\label{diagonal}
 b_{i,i}= 4\sum_{k=1}^{m-1}(x_{k,i}-x_{k+1,i})^2,
\end{equation}
for $1\leq i\leq d$ and the off-diagonal entries as
\begin{equation}\label{off-diagonal}
 b_{i,j}=4 \sum_{k=1}^{m-1}(x_{k,i}-x_{k+1,i})(x_{k,j}-x_{k+1,j}),
\end{equation}
for $1\leq i\neq j\leq d$ where $x_{k,i}$ is the $i$-th element of vector $x_k$. In the following lemmas, we show that with high probability, as $m$ increases, the diagonal entries of $B$ will all be of the order of $m$, i.e., $b_{i,i}=\Theta(m)$, and the off-diagonal entries will be bounded from above by $m^{\frac{1}{2}+\epsilon}$, i.e., $b_{i,j}=o(m)$.

\begin{lemma}\label{diagonal-concentration}
For any  $\epsilon>0$ the diagonal entries of $B$ are bounded as follows.
\begin{displaymath}
 \prob\left(|b_{i,i}-2(m-1)/3|>4 m^{\frac{1}{2}+\epsilon}\right)\leq 4 e^{-m^{2\epsilon}}.
\end{displaymath}
\end{lemma}
%\begin{proof}
 The idea is to use Hoeffding's Inequality (see appendix \ref{Hoeffding}) for the sum of independent and bounded random variables. To this end, we need to divide the sum in \eqref{diagonal} into sums of even and odd terms as follows:
\begin{displaymath}
 b_{i,i}=b_e^i+b_o^i,
\end{displaymath}
where
\begin{eqnarray}
b_e^i &=& 4\sum_{k\in \textrm{even}}(x_{k,i}-x_{k+1,i})^2,  \label{even-term}\\
b_o^i &=& 4\sum_{k\in \textrm{odd}}(x_{k,i}-x_{k+1,i})^2.\label{odd-term}
\end{eqnarray}
This separation ensures that the random variables in summations \eqref{even-term} and \eqref{odd-term} are independent. Let the random variable $z_{k}^i$ denote the term $4(x_{k,i}-x_{k+1,i})^2$ in \eqref{even-term}. Since $z^i_k\in[0,4]$ and all the terms in $b^i_e$ are independent of each other, we can use  Hoeffding's Inequality to upper bound the probability of the deviation of $b^i_e$ from its expected value:
\begin{equation}\label{dev1}
 \prob\left(|b^i_e-(m-1)/3|>2 m^{\frac{1}{2}+\epsilon}\right)\leq 2 e^{-m^{2\epsilon}},
\end{equation}
for any fixed $\epsilon>0$.
The same bound holds for $b_o$. Namely,
\begin{equation}\label{dev2}
\prob\left(|b^i_o-(m-1)/3|>2 m^{\frac{1}{2}+\epsilon}\right)\leq 2 e^{-m^{2\epsilon}}.
\end{equation}
Hence,
\begin{small}
\begin{align*}
&  \prob\left(|b_{i,i}-2(m-1)/3|>4 m^{\frac{1}{2}+\epsilon}\right) \\
&\quad\quad \stackrel{(a)}{\leq} \prob\left(|b_e-(m-1)/3|+|b_o-(m-1)/3|>4 m^{\frac{1}{2}+\epsilon}\right)\\
%&\quad\quad\stackrel{(b)}{\leq} \prob\left(|b_e-(m-1)/3|>2 m^{\frac{1}{2}+\epsilon}\right)+ \prob\left(|b_o-(m-1)/3|>2 m^{\frac{1}{2}+\epsilon}\right)\\
&\quad\quad\stackrel{(b)}{\leq} 4 e^{-m^{2\epsilon}},
\end{align*}
\end{small}
where in $(a)$ we used triangular inequality and in $(b)$ we used the union bound.
%\end{proof}
\begin{lemma}\label{off-diagonal-concentration}
For any  $\epsilon>0$ the off-diagonal entries of $B$ are bounded as follows.
\begin{displaymath}
 \prob\left(|b_{i,j}|>16 m^{\frac{1}{2}+\epsilon}\right)\leq 4 e^{-m^{2\epsilon}}.
\end{displaymath}

\end{lemma}
%\begin{proof}
 The proof follows in the same lines as the proof of Lemma~\ref{diagonal-concentration}. 
%\end{proof}

Using the Gershgorin circle theorem (see appendix \ref{Gershgorin}) we can find a lower bound on the minimum eigenvalue of $B$.
\begin{equation}\label{lambda-min}
 \lambda_\textrm{min}(B)\geq \min_i (b_{i,i}-R_i), 
\end{equation}
where
\begin{displaymath}
 R_i=\sum_{j\neq i }|b_{i,j}|.
\end{displaymath}
Now, let $\mathbb{B}_{ii}$ denote the event that $\{ b_{i,i}<2(m-1)/3-4m^{\frac{1}{2}+\epsilon}\}$ and $\mathbb{B}_{ij}$ (for $i\neq j$)  denote the event that $\{b_{i,j}>16m^{\frac{1}{2}+\epsilon}\}$. Since the matrix $B$ is symmetric, we  have only $d(d+1)/2$ degrees of freedom. Lemma \ref{diagonal-concentration} and \ref{off-diagonal-concentration} provide us with a bound on the probability of each event. Therefore, by using the union bound we get
\begin{eqnarray*}
 \prob\left(\bigcup_{i\leq j} \overline{\mathbb{B}_{ij}} \right)&\leq&1- \sum_{i\leq j}\prob(\mathbb{B}_{ij})\\ 
&=&1- 3d^2 e^{-m^{2\epsilon}}. 
\end{eqnarray*}
Therefore with probability at least $1-3d^2 e^{-m^{2\epsilon}}$ we have 
\begin{equation}\label{bii-lower}
 b_{i,i}-R_i\geq \frac{2(m-1)}{3}-16 d\cdot m^{\frac{1}{2}+\epsilon},
\end{equation}
for all $1\leq i\leq d$. As $m$ grows, the RHS of \eqref{bii-lower} can be lower bounded by $(m-1)/3$. By combining \eqref{lambda-min} and \eqref{bii-lower} we can conclude that
\begin{equation}\label{lambda-lower}
 \prob \left(\lambda_\textrm{min}(B)\geq\frac{(m-1)}{3} \right)\geq1-3d^2 e^{-m^{2\epsilon}}.
\end{equation}
As a result, from \eqref{sigma-min} and \eqref{lambda-lower} we have
\begin{equation}
 \prob\left(\|(A^TA)^{-1}A\|_2 \leq \sqrt{\frac{3}{m-1}}\right)\geq 1-3d^2 e^{-m^{2\epsilon}},
\end{equation}
which shows that  as $m$ grows, with high probability we have $\|(A^TA)^{-1}A\|_2\leq\sqrt{\frac{3}{m-1}}$.

%
%============================================================
%
\subsection{Proof of Lemmas \ref{lem:spectraldeterministic} (Deterministic Model) }
\label{sec:spectral}

By using the singular value decomposition of 
a tall $m-1 \times d$ matrix $A$, we know that it can be written as 
$A=U\Sigma V^T$ where $U$ is an orthogonal matrix, $V$ is a unitary matrix and $\Sigma$ is a diagonal matrix. Then, 
\begin{displaymath}
(A^TA)^{-1}A=U\Sigma^{-1} V^T.
\end{displaymath}
Hence,
\begin{equation}\label{sigma-min}
 \|(A^TA)^{-1}A\|_2=\frac{1}{\sigma_{min}(A)},
\end{equation}
where $\sigma_{min}(A)$ is the smallest singular value of $A$. This means that in order to upper bound $\|(A^TA)^{-1}A\|_2$ we need to lower bound the smallest singular value of $A$. 

By putting the sensors in the mentioned positions the $d\times d$ matrix $A$ will be Toeplitz and have the following form.
\begin{displaymath}
 A=2\left[\begin{array}{c c c c c}
1 & -1 & 0 & \cdots & 0 \\
0 & 1  & -1 & \cdots & 0 \\
\vdots & \vdots & \ddots & \ddots & \vdots\\
  0  & \cdots & 0 & 1 & -1 \\
  0  & \cdots & 0 & 0 & 1
\end{array} \right].
\end{displaymath}
We can easily find the inverse of matrix $A$.
\begin{displaymath}
 A^{-1}=\frac{1}{2}\left[\begin{array}{c c c c c}
1 & 1 & 1 & \cdots & 1 \\
0 & 1  & 1 & \cdots & 1 \\
\vdots & \vdots & \ddots & \ddots & \vdots\\
  0  & \cdots & 0 & 1 & 1 \\
  0  & \cdots & 0 & 0 & 1
\end{array} \right].
\end{displaymath}
Note that the maximum singular value of $A^{-1}$ and the minimum singular value of $A$ are related as follows.
\begin{equation}
 \sigma_{min}(A)=\frac{1}{\sigma_{max}(A^{-1})}. 
\end{equation}
To find the maximum singular value of $A^{-1}$, we need to calculate the maximum eigenvalue of $A^{-1}\left(A^{-1}\right)^T$ which has the following form
 \begin{displaymath}
 A^{-1}\left(A^{-1}\right)^T=\frac{1}{4}\left[\begin{array}{c c c c c}
d & d-1 & d-2 & \cdots & 1 \\
d-1 & d-1  & d-2 & \cdots & 1 \\
\vdots & \vdots & \ddots & \ddots & \vdots\\
  2  & \cdots & 2 & 2 & 1 \\
  1  & \cdots & 1 & 1 & 1
\end{array} \right].
\end{displaymath}
By using the Gershgorin circle theorem (see appendix \ref{Gershgorin}) we can find an upper bound on the maximum eigenvalue of $A^{-1}\left(A^{-1}\right)^T$.
\begin{equation}\label{lambda-min-det}
 \lambda_\textrm{max}\left(A^{-1}\left(A^{-1}\right)^T\right)\leq \frac{d^2}{4}, 
\end{equation}
Hence, by combining \eqref{sigma-min} and \eqref{lambda-min-det} we get
\begin{equation}
 \|(A^TA)^{-1}A\|_2\leq \frac{d}{2}.
\end{equation}

%%%%%%%%%%%%%%%%%%%%%%%%%%%%%%%%%%%%%%%%%%%%%%%%%%%%%%%%%%%%%%%%%%%%%%%%%%%%%%%%

\begin{comment}
\section{Conclusion}
In many applications of wireless sensor networks, it is crucial to determine the location of nodes. Distributed localization of nodes is a key to enable most of these applications. For this matter, numerous algorithms have been recently proposed where the efficiency and success of them have been mostly demonstrated by simulations. In this paper, we investigated the distributed sensor localization problem from a theoretical point of view and provided analytical bounds on the performance of such an algorithm. More precisely, we analyzed  the {\sc HOP-TERRAIN} algorithm and  showed that even when only the connectivity information was given, the Euclidean distance between the estimate and the correct position of every unknown node is bounded and decays at a rate inversely proportional to the radio range. In the case of noisy distance measurements, we observe the same behavior and a similar bound holds. 
%A possible future research direction is to extend the proposed results with uniform distribution over $[0,1]^d$ to study some real network topologies.
\end{comment}

%
%============================================================
%
%\appendix

\subsection{Proof of the Bound on the Number of Hops}
\label{sec:hops}

We start by applying a bin-covering technique in a similar way as in \cite{MP05,OKM10,KO10}. 
In this section, for simplicity, we assume that the nodes are placed in a $3$-dimensional space. 
However, analogous argument proves that the same statement is true for $d=2$ as well. 
%\begin{figure}[t]
%\vspace{-1cm}
%	\begin{center}
%	\includegraphics[width=9cm]{dist_shortestpath}
%	%\vspace{-2cm}
%	\caption{The shortest path between the two red nodes is upper bounded by the path going through the blue nodes. Under the right scaling, we can ensure that with high probability there exists one node inside each square. This way, we can bound the shortest path, by adding up the number of squares between the two red nodes. } 
%	\label{fig:hops}
%	\end{center}
%\end{figure}

For each ordered pair of nodes $(i,j)$ such that $d_{i,j}>R$, define a `bin' as 
\begin{eqnarray*}
 A_{i,j} = \big\{ x\in[0,1]^3 \;\big|\; R-\delta\leq d(x,x_i)\leq R, \measuredangle(x_j-x_i,x-x_i)\leq\theta \big\} \;,
\end{eqnarray*}
where $\delta$ and $\theta$ are positive parameters to be specified later in this section, 
and $\measuredangle(\cdot,\cdot):\reals^d\times\reals^d\to[0,\pi]$ is the angle between two vectors:  
$$\measuredangle(\cdot,\cdot)\equiv \arccos(z_1^Tz_2/(\|z_1\|\|z_2\|)).$$ 
We say a bin $A_{i,j}$ is occupied if there is a node inside the bin that is detected by node $i$ 
(i.e., conencted to node $i$ in the graph $G$). 
Next, for each unordered pair of nodes $(i,j)$ such that $d_{i,j}\leq R$, define a bin as 
\begin{eqnarray*}
 B_{i,j} = \big\{ x\in[0,1]^3 \;\big|\; d(x,x_i)\leq R, d(x,x_j)\leq R \big\} \;.
\end{eqnarray*}
We say a bin $B_{i,j}$ is occupied if there is a node inside the bin 
that is simultaneously detected by nodes $i$ and $j$  
(i.e., connected to both nodes $i$ and $j$ in the graph $G$). 
When $n$ nodes are deployed in $[0,1]^d$ uniformly at random, 
we want to ensure that, with high probability, 
all bins are occupied for appropriate choices of $R$, $\delta$, and $\theta$. 

First when $d_{i,j}>R$, 
\begin{eqnarray*}
  \prob\big(A_{i,j} \text{ is occupied} \big) &=& 1-\prod_{l\neq i,j} (1-\prob(\text{node $l$ occupies $A_{i,j}$})) \\
&\geq& 1-\left(1-\frac{1}{4}\int_{0}^{\theta}\int^{R}_{R-\delta} 2\pi r^2\sin(\phi)p(r) \mathrm{d}r \mathrm{d}\phi  \right)^{n-2} \\
	&=& 1-\left(1-\frac{1}{2}\pi\alpha (1-\cos(\theta)) R^\beta \frac{1}{3-\beta} (R^{3-\beta}-(R-\delta)^{3-\beta})\right)^{n-2}\;, 
\end{eqnarray*}
for $\beta\in[0,3)$ as per our assumption. 
Since $A_{i,j}$'s are constrained to be in $[0,1]^3$, we need to scale the probability by $1/4$.   
The above inequality is tight in the worst case, that is  when 
both nodes $i$ and $j$ lie on one of the edges of the cube $[0,1]^3$. 
%, and scaled the probability of occupancy by $1/2^d$. 
We choose $\theta$ such that $1-\cos(\theta)=(\delta/R)^2$. 
Then using the facts that $1-z\leq\exp(-z)$ and 
$(1-z^{3-\beta})\leq (3-\beta)(1-z)/3$ for $z\in[0,1)$ and $\beta\in[0,3)$, we have 
\begin{eqnarray}
  \label{eq:Aoccupy}
  \prob\big(A_{i,j} \text{ is occupied}\big) \geq 1-\exp\left(-\frac{\pi\alpha \delta^3}{6}(n-2)\right)\;, 
\end{eqnarray}
which is larger than $1-1/n^6$ if we set 
$\delta = (12\log n/(\alpha(n-2)))^{1/3}$. 
%An analogous argument proves a similar result for $d=2$.

Next we consider the case when nodes $i$ and $j$ are at most $R$ apart. 
Notice that nodes $i$ and $j$ may not be directly connected in the graph $G$, 
even if they are within a radio range $R$.
The probability that they are not directly connected is 
$1-\alpha(d_{i,j}/R)^{-\beta}$, which does not vanish even for large $n$.
However, we can show that nodes $i$ and $j$ are at most $2$ hops apart with overwhelming probability. 
The event that $h_{i,j}>2$ is equivalent to the event that $B_{i,j}$ is occupied. 
Then, 
\begin{eqnarray}
 \prob\big(B_{i,j} \text{ is occupied} \big) &=& 1-\prod_{l\neq i,j} (1-\prob(\text{node $l$ is detected by $i$ and $j$})) \nonumber\\
  		    &\geq& 1-(1-V(B_{i,j})\alpha^2)^{n-2} \nonumber\\
		    &\geq& 1-\exp\big\{{-V(B_{i,j})\alpha^2(n-2)}\big\} \label{eq:singlehop}\;,
\end{eqnarray}
where $V(B_{i,j})\in\reals$ is the volume of $B_{i,j}$,  
%the set $A_d(d_{i,j})=\{x\in\reals^d\,|\,d(x,x_i)\leq R,d(x,x_j)\leq R\}$, 
and we used the fact that the probability of detection is lower bounded by $\alpha$. 
$V(B_{i,j})$ is the smallest when nodes $i$ and $j$ are distance $R$ apart and 
lie on one of the edges of the cube $[0,1]^3$.
%The volume $V_d(d_{i,j})$ is a monotonically decreasing function of $d_{i,j}$, which implies $V_d(d_{i,j})\geq V_d(R)$. 
%It is therefore sufficient to lower bound $V_d(R)$ to provide an upper bound on \eqref{eq:singlehop}. 
%To provide an upper bound on \eqref{eq:singlehop}. 
%In a $2$-dimensional space, the set $A_2(R)$ covers two identical triangles of height $R/2$ and base length $\sqrt{3}R$. 
%This gives a lower bound $V_2(R)\geq(\sqrt{3}/2)R^2$. 
%In a $3$-dimensional space, the set $A_3(R)$ 
%covers two identical circular cones of height $R/2$ and base area $(3\pi/4)R^2$. 
%This gives a lower bound $V_3(R)\geq(\pi/8)R^3$. 
%In a general $d$ dimensional space, the set $A_d(R)$ covers two identical cones of height $R/2$ and with 
%$(d-1)$-dimensional ball of radius $(\sqrt{3}/2)R$ as a base. 
%Let $\Gamma(\cdot)$ denote the gamma function:  
%$\Gamma(n)=(n-1)!$ and $\Gamma(n+(1/2))=\sqrt{\pi}(2n)!/(4^n n!)$.
% Since the volume of a $(d-1)$-ball of radius $R$ is $(\pi^{(d-1)/2}R^{d-1})/\Gamma((d+1)/2)$, 
% we have 
% \begin{eqnarray} 
%  \label{eq:volume}
%  V_d(R)\geq C_d\,R^d \;,
% \end{eqnarray}
% where $C_d=\frac{\pi^{(d-1)/2}}{d\;\Gamma((d+1)/2))}$ is a constant that only depends on $d$.
%In a $2$-dimensional space, $V(B_{i,j})\geq (1/2) \big((2/3)\pi-(\sqrt{3}/4)\big)R^2 \geq (1/2^{d-1})R^2$. 
In a $3$-dimensional space, $$V(B_{i,j})\geq (1/4)(5/12)\pi R^3 \geq (1/4)R^3.$$ 
Substituting these bounds in \eqref{eq:singlehop}, we get 
\begin{eqnarray} 
 \prob\big(B_{i,j} \text{ is occupied} \big) \geq 1-\exp\big\{{- (1/4) \alpha^2 R^3 (n-2)}\big\} \;, \label{eq:Boccupy}
\end{eqnarray}
which is larger than $1-1/n^6$ for $R\geq\big((24\log n)/((n-2)\alpha^2)\big)^{1/3}$. 

For each ordered pair $(i,j)$, we are interested in the bin $A_{i,j}$ if $d_{i,j}>R$ 
and $B_{i,j}$ if $d_{i,j}\leq R$. 
Using the bounds in \eqref{eq:Aoccupy} and \eqref{eq:Boccupy}  
and applying union bound on all $n(n-1)$ ordered pairs of nodes, 
all bins $$\{A_{i,j}\,|\,d_{i,j}>R\}\cup \{B_{i,j}\,|\,d_{i,j}\leq R\}$$ 
are occupied with a probability larger than $1-1/n^4$. 

Now assuming all bins are occupied, we first show that 
the number of hops between two nodes $i$ and $j$ is 
bounded by a function $F(d_{i,j})$ that only depends on the distance
between the two nodes. 
The function $F:\R^+\rightarrow\R^+$ is defined as 
\begin{eqnarray*}
 F(z) = \left\{  \begin{array}{ ll}
       	2      & \text{if } z \leq R\;, \\
       	k+2 & \text{if } z \in  \cL_k\; \text{ for } k\in\{1,2,\ldots\}\;,
	\end{array}
 	\right. 
\end{eqnarray*}
where $\cL_k$ denotes the interval $(k(R-\sqrt{3}\delta)+\sqrt{3}\delta,k(R-\sqrt{3}\delta)+R]$.
Our strategy is to use induction to show that for all pairs of nodes, 
\begin{eqnarray}
 h_{i,j} \leq F(d_{i,j}) \;. \label{eq:boundF}
\end{eqnarray}

First, assume nodes $i$ and $j$ are at most $R$ apart. 
Then, by the assumption that the bin $B_{i,j}$ is occupied  
there is a node connected to both $i$ and $j$.  
Therefore the number of hops $h_{i,j}$ is at most $2$. 

Next, assume that the bound in \eqref{eq:boundF} is true for 
all pairs $(l,m)$ with $$d_{l,m}\leq \sqrt{3}\delta + k(R-\sqrt{3}\delta).$$
For two nodes $i$ and $j$ at distance $d_{i,j} \in \cL_k$, 
consider a line segment $\ell_{i,j}$ in the $3$-dimensional space 
with one end at $x_i$ and the other at $x_j$.
Let $y\in\reals^3$ be the point in the line segment $\ell_{i,j}$ that is at distance $R$ from $x_i$.
We want to show that there exists a node that is close to $y$ and is connected to node $i$.
By definition, $y$ is inside the bin $A_{i,j}$. 
We know that the bin $A_{i,j}$ is occupied by at least one node that is connected to node $i$. 
Let us denote one of these nodes by $l$. 
Then $d(y,x_l)\leq \sqrt{3}\delta$ because  
$$\sup_{z \in A_{i,j}} d(z,y)=\sqrt{\delta^2+2R(R-\delta)(1-\cos(\theta))}\leq \sqrt{3}\delta.$$

We use the following triangular inequality which follows from the definition of the number of hops. 
\begin{eqnarray*}
    h_{i,j} \leq h_{i,l} + h_{l,j} \;. 
\end{eqnarray*}
Since $l$ is connected to $i$ we have $h_{i,l}=1$. 
By triangular inequality, we also have $d_{l,j} \leq d(y,x_j)+d(y,x_l)$. 
It follows from $d(y,x_j)=d_{i,j}-R$ and $d(y,x_l)\leq \sqrt{3}\delta$ that $$d_{l,j} \leq d_{i,j}-R+\sqrt{3}\delta.$$
Recall that we assumed $d_{i,j}\leq R+k(R-\sqrt3\delta)$.
Since we assumed that \eqref{eq:boundF} holds for $ d_{l,j}\leq \sqrt3\delta+k(R-\sqrt{3}\delta)$, 
we have 
\begin{eqnarray*}
 h_{i,j}\leq k+2\;, 
\end{eqnarray*}
for all nodes $i$ and $j$ such that $d_{i,j}\leq R + k(R-\sqrt{3}\delta)$. 
By induction, this proves that the bound in \eqref{eq:boundF} holds for all pairs $(i,j)$. 

We can upper bound $F(z)$ with a simple affine function: 
\begin{eqnarray*}
 F(z) &\leq& 2 + \frac{1}{R-\sqrt{3}\delta}z\\
	&\leq& 2 + \Big(1+\frac{2\delta}{R}\Big)\frac{z}{R} \;, 
\end{eqnarray*}
where the last inequality is true for $R\geq2\sqrt{3}\delta/(2-\sqrt{3})$. 
Together with \eqref{eq:boundF} this finishes the proof of the lemma. 

\begin{figure}
\begin{center}
%\hspace{-2.1cm}
\includegraphics[width=8cm]{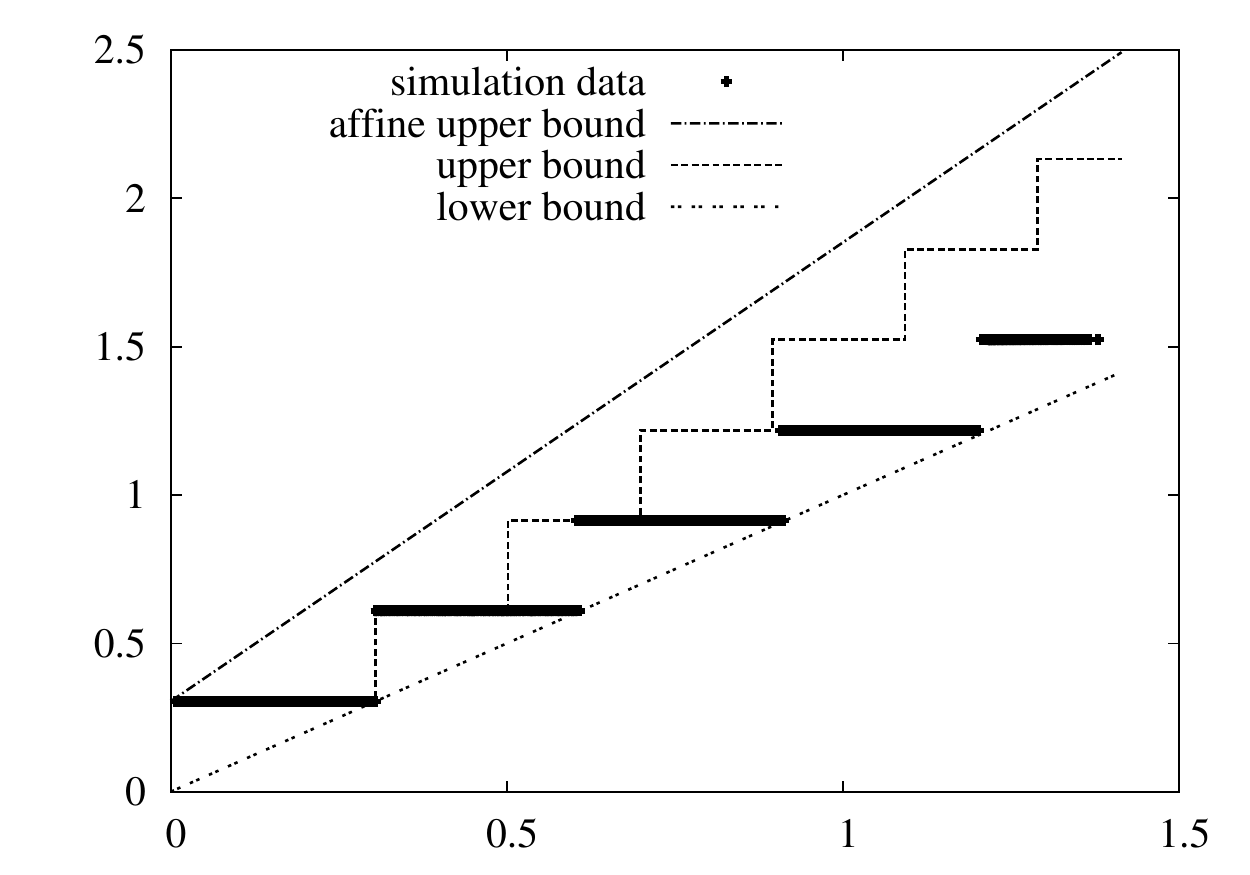}
\put(-105,-5){$d_{i,j}$}
\put(-225,80){$\hd_{i,j}$}
%\put(-345,90){\begin{sideways}\scriptsize{Relative Error}\end{sideways}}
\caption{Comparison of upper and lower bound of shortest paths $\{\hd_{i,j}\}$ with respect to the correct distance $\{d_{i,j}\}$ computed for $n=6000$ sensors in $2$-dimensional square $[0,1]^2$ under connectivity-based model.} \label{fig:shortestpath_connectivity}
%\vspace{-1.cm}
\end{center}
\end{figure}
\looseness -1
Figure \ref{fig:shortestpath_connectivity} illustrates 
the comparison of the 
upper bounds $F(d_{i,j})$ and $F_a(d_{i,j})$, 
and the trivial lower bound $\hd_{i,j}\geq d_{i,j}$
in a simulation with parameters $d=2$, $n=6000$ and 
$R=\sqrt{64 \log n/n}$.
The simulation data shows the distribution of shortest paths 
between all pairs of nodes with respect to the actual pairwise distances,
which confirms that the shortest paths lie between 
the analytical upper and lower bounds.
Although the gap between the upper and lower bound is seemingly large, 
in the regime where $R=C\sqrt{\log n/n}$ with a constant $C$,
the vertical gap $R$ vanishes as $n$ goes to infinity 
and the slope of the affine upper bound can be made 
arbitrarily small by increasing the radio range $R$ 
or equivalently taking large enough $C$.
\looseness -1
\section{Conclusion}
\label{sec:conclusion}
In many applications of wireless sensor networks, it is crucial to determine the location of nodes. 
For this matter, numerous algorithms have been recently proposed where the efficiency and success of 
them have been mostly demonstrated by simulations. In this paper, we have investigated the centralized 
and distributed sensor localization problem from a theoretical point of view and have provided analytical 
bounds on the performance of such  algorithms. More precisely,
 we analysed  the {\sc MDS-MAP} and {\sc HOP-TERRAIN} algorithms and 
showed that even when only the connectivity information was given and 
in the presence of detection failure, the resulting error of both algorithms 
is bounded and decays at a rate inversely proportional to the detection range. 

\section*{Acknowledgment}
We would like to thank Andrea Montanari and R\"udiger Urbanke for their stimulating discussions on the subject of this paper.

\bibliographystyle{alpha}
\bibliography{positioning}

\appendix

\section{Hoeffding's Inequality}\label{Hoeffding}

\textbf{Hoeffding's inequality} \cite{Hoe63} is a result in probability theory that gives an upper bound on the probability for the sum of random variables to deviate from its expected value. Let $z_1, z_2, \dots, z_n$ be independent and bounded random variables such that $z_k\in [a_k,b_k]$ with probability one. Let $s_n=\sum_{k=1}^nz_k$. Then for any $\delta>0$, we have 
\begin{displaymath}
 \prob\left(|s_n-\mathbb{E}[s_n]|\geq \delta\right)\leq 2 \exp\left(-\frac{2\delta^2}{\sum_{k=1}^n(b_k-a_k)^2}\right).
\end{displaymath}

\section{Gershgorin circle theorem}\label{Gershgorin}

The \textbf{Gershgorin circle theorem} \cite{Horn85} identifies a region in the complex plane that contains all the eigenvalues of a complex square matrix. For an $n\times n$ matrix $A$, define
\begin{displaymath}
 R_i=\sum_{j\neq i }|a_{i,j}|.
\end{displaymath}
Then each eigenvalue of $A$ is in at least one of the disks
\begin{displaymath}
 \{z:|z-a_{i,i}|\leq R_i\}.
\end{displaymath}

\section{Concentration of Lipschitz functions}
Informally, \textbf{concentration of Lipschitz functions} says that any \textit{smooth} function of bounded
independent random variables is tightly concentrated around its expectation \cite{measure}. The notion of smoothness we
will use is \textit{Lipschitz}.
\begin{definition} $f:\R^n\rightarrow R$ is $\lambda$-Lipschitz with respect to the $l_p$ norm, if for all $x$ and $y$, 
$$|f(x)-f(g)|\leq \lambda\|x-y\|_p.$$
\end{definition}
It turns out that Hoeffding's bound holds for all Lipschitz (with respect to $l_1$ norm) functions. More precisely, suppose $X_1,X_2, \dots, X_n$ are independent and bounded with $a_i\leq x_i\leq b_i$. Then for any $f:\R^n\rightarrow R$ which is $\lambda$-Lipschitz with respect to the $l_1$ norm, $$\Pr(|f-E(f)|\geq \epsilon)\leq 2\exp\left(-\frac{2\epsilon^2}{\lambda^2\sum_{i=1}^n(b_i-a_i)^2}\right).$$

%
%============================================================
%\diamond

%
%============================================================
%
%\bibliographystyle{amsalpha}
%\bibliographystyle{IEEEtran}

%\bibliography{sigproc}

% that's all folks

\end{document}